\documentclass[11pt]{article}
\pdfoutput=1

\RemoveFromHook{package/amsthm/after}[firstaid/aliascounter]
\usepackage{amssymb,amsmath,amsthm,amsfonts}
\usepackage[utf8]{inputenc}
\usepackage[T1]{fontenc}
\usepackage[tt=false, type1=true]{libertine}
\usepackage[varqu]{zi4}
\usepackage[libertine]{newtxmath}
\usepackage[margin=1in]{geometry}
\usepackage{hyperref}
\hypersetup{
    colorlinks,
    allcolors=teal
}

\usepackage[sortcites,sorting=nyt,style=alphabetic,maxbibnames=7]{biblatex}
\usepackage[ruled]{algorithm2e}
    
    \SetAlFnt{\small}
    \SetAlCapFnt{\small}
    \SetAlCapNameFnt{\small}
    \SetAlCapHSkip{0pt}
    \IncMargin{-\parindent}

\usepackage{tikz, xifthen, physics, xcolor, suffix, nicefrac, enumitem, subcaption, float, thm-restate, wrapfig, mathtools}
\usetikzlibrary{shapes, arrows, positioning}
\usepackage[bb=dsserif]{mathalpha}
\usepackage[capitalize]{cleveref}

\renewcommand{\paragraph}[1]{\smallskip\noindent\textbf{#1.}}
\newcommand{\itparagraph}[1]{\textit{#1.}}

\newenvironment{optimization}{\crefalias{equation}{optimization}\equation\aligned}{\endaligned\endequation}
\newenvironment{optimization*}{\displaymath\aligned}{\endaligned\enddisplaymath}

\crefname{optimization}{Optimization Program}{Optimization Programs}
\creflabelformat{optimization}{(#2#1#3)}

\DeclareMathOperator*{\E}{\mathbb{E}}
\DeclareMathOperator*{\PR}{\mathbb{P}}
\newcommand{\Ex}[2][]{\E_{\substack{#1}}\left[\mathchoice{\big.}{}{}{}#2\right]}

\renewcommand{\Pr}[2][]{\PR_{\substack{#1}}\left[\mathchoice{\big.}{}{}{}#2\right]}

\newcommand{\One}[1]{\mathbb{1}_{#1}}
\newcommand{\orderT}[1]{\tilde{\mathcal{O}}\qty(#1)}
\WithSuffix\newcommand\orderT*[1]{\tilde{\mathcal O}(#1)}
\DeclareMathOperator*{\argmax}{arg\;\!max\,}
\DeclareMathOperator*{\argmin}{arg\;\!min\,}
\newcommand{\m}[1]{\begin{bmatrix}#1\end{bmatrix}}

\newcommand{\R}{\mathbb{R}}

\newcommand{\Rz}{\R_{\ge 0}}

\newcommand{\calA}{\mathcal{A}}
\newcommand{\calB}{\mathcal{B}}

\newcommand{\floor}[1]{\left\lfloor{#1}\right\rfloor}
\WithSuffix\newcommand\floor*[1]{\lfloor{#1}\rfloor}
\newcommand{\ceil}[1]{\left\lceil{#1}\right\rceil}
\WithSuffix\newcommand\ceil*[1]{\lceil{#1}\rceil}

\let\e\varepsilon
\let\l\lambda
\let\d\delta
\let\sub\subseteq

\newcommand{\SE}{\mathtt{SE}}

\newcommand{\conv}{\textrm{cnvx}}
\newcommand{\opt}{\mathtt{OPT}}
\newcommand{\hatvO}{\hat{v}_{\smash{\mathtt{O}}}}
\newcommand{\hatVO}{\hat{V}_{\smash{\mathtt{O}}}}

\newcommand{\vO}{v_{\smash{\mathtt{O}}}}
\newcommand{\vL}{v_{\smash{\mathtt{L}}}}
\newcommand{\UO}{U_{\smash{\mathtt{O}}}}
\newcommand{\UL}{U_{\smash{\mathtt{L}}}}
\newcommand{\PO}{P_{\smash{\mathtt{O}}}}
\newcommand{\PL}{P_{\smash{\mathtt{L}}}}
\newcommand{\rO}{\rho_{\smash{\mathtt{O}}}}
\newcommand{\rL}{\rho_{\smash{\mathtt{L}}}}
\newcommand{\BR}{\textrm{BR}}

\newcommand{\gs}{g^\star}

\newcommand{\Line}[4]{%
    #1&%
    \ifthenelse{\isempty{#2}}{\phantom{=}}{#2}%
    #3%
    \ifthenelse{\isempty{#4}}{}{&&\qquad\left(\big.\substack{#4}\right)}%
}

\newtheorem{theorem}{Theorem}[section]

\newtheorem{lemma}[theorem]{Lemma}
\newtheorem{claim}[theorem]{Claim}

\theoremstyle{definition}
\newtheorem{remark}[theorem]{Remark}

\newtheorem{definition}[theorem]{Definition}

\crefname{property}{Property}{Properties}

\newtheorem*{fact*}{Fact}
\crefname{fact}{Fact}{Facts}

\crefname{assumption}{Assumption}{Assumptions}

\newcommand{\citet}[1]{\textcite{#1}}

\title{Learning vs. Optimizing Bidders in Budgeted Auctions}

\author{
Giannis Fikioris\\
Google Research%
\footnote{Part of this work was done while a PhD student at Cornell, supported by the Google PhD Fellowship and ONR MURI grant N000142412742.}\\
\texttt{giannisf@google.com}
\and
Balasubramanian Sivan\\
Google Research\\
\texttt{balusivan@google.com}
\and
\'Eva Tardos%
\thanks{Supported in part by AFOSR grant FA9550-23-1-0410, AFOSR grant FA9550-231-0068, ONR MURI grant N000142412742, and through a SPROUT Award from Cornell Engineering.}\\
Cornell University \& Google Research\\
\texttt{eva.tardos@cornell.edu}
}

\date{}

\begin{document}

\pagenumbering{roman}
\maketitle{}
\thispagestyle{empty}

\begin{abstract}
    The study of repeated interactions between a learner and a utility-maximizing optimizer has yielded deep insights into the manipulability of learning algorithms. However, existing literature primarily focuses on independent, unlinked rounds, largely ignoring the ubiquitous practical reality of budget constraints. In this paper, we study this interaction in repeated second-price auctions in a Bayesian setting between a learning agent and a strategic agent, both subject to strict budget constraints, showing that such cross-round constraints fundamentally alter the strategic landscape.

First, we generalize the classic Stackelberg equilibrium to the \emph{Budgeted Stackelberg Equilibrium}. We prove that an optimizer's optimal strategy in a budgeted setting requires time-multiplexing; for a $k$-dimensional budget constraint, the optimal strategy strictly decomposes into up to $k+1$ distinct phases, with each phase employing a possibly unique mixed strategy (the case of $k=0$ recovers the classic Stackelberg equilibrium where the optimizer repeatedly uses a single mixed strategy). Second, we address the intriguing question of non-manipulability. We prove that when the learner employs a standard Proportional controller (the ``P'' of the PID-controller) to pace their bids, the optimizer's utility is upper bounded by their objective value in the Budgeted Stackelberg Equilibrium baseline. By bounding the dynamics of the PID controller via a novel analysis, our results establish that this widely used control-theoretic heuristic in an important class of repeated games with budgets is actually strategically robust.

\end{abstract}

\clearpage
\pagenumbering{arabic}
\section{Introduction} \label{sec:intro}

The study of online learning in strategic environments has blossomed over the past decade into a fertile and dynamic area of theoretical computer science. A recent central theme of this literature is understanding the long-run outcomes when a standard, no-regret learning algorithm interacts with a fully rational, utility-maximizing optimizer. In these asymmetric interactions, prior work has consistently asked: What extra utility can an optimizer extract when playing against a predictable learner, rather than a strategic opponent? Which learning algorithms are fundamentally robust against such manipulation? 

This inquiry has yielded a rich understanding of optimizer-learner dynamics. We know that against standard no-regret learners, optimizers can always earn an average objective of their Stackelberg value in the single-shot game~\cite{DSS19} (by simply playing their Stackelberg strategy in each round), and often much more than this. Preventing earning more via sophisticated manipulation requires learners to minimize stronger notions of regret---linearized regret in auctions \cite{KSS24, cai2026onlinelinearoptimizationsufficient, DBLP:journals/corr/abs-2601-03853}, swap regret in normal-form games~\cite{DSS19, MMSS22}, and profile swap regret in broader polytope games~\cite{ACMMSS25}. However, all of this foundational work operates under a simplifying assumption: decisions are made round-by-round, independent of the past.

In this paper, we add to this line of work a new, practically ubiquitous, dimension: \emph{budget constraints}. In real-world repeated interactions---most notably in digital advertising and online auctions~\cite{GooglePacing, TwitterPacing}---agents rarely operate with infinite liquidity. Bidding under budget constraints is a thriving area of research in its own right, precisely because budgets fundamentally alter strategic behavior, even in traditionally truthful mechanisms like second-price auctions. In a further departure from prior work, which has overwhelmingly modeled the optimizer as an unconstrained \emph{seller} (e.g., dynamically setting reserve prices to extract revenue: see~\cite{BMSW18, KSS24, GJLM23}) and the learner as a \emph{buyer}, we shift this paradigm to study a setting where \emph{both} the optimizer and the learner are competing \emph{buyers}, with budget constraints. 
We study the interaction between two budget-constrained agents participating in repeated auctions (we primarily focus on second-price auctions, and also show, under some assumptions, that our results extend to first-price auctions), where item values are drawn repeatedly from a joint distribution. One agent employs a learning algorithm to compute their bids, while the opponent is a strategic optimizer. Both share the same objective: maximizing their total accumulated value subject to a budget constraint.

Budgets introduce a significant wrinkle---by linking decisions across rounds, budgets shatter the isolation of the repeated game framework. First, from the optimizer's perspective, it is unclear what the generalization of the Stackelberg strategy is when they are budget constrained. Second, from the learner's perspective, it is unclear whether they can achieve standard properties like no-regret, and best respond to the optimizer's strategy. The learner being able to do as good as their best distribution in hindsight becomes complicated because of the ``spend or save'' dilemma introduced by the budget constraints (see~\cite{ISSS22}), that forces the learner to weigh the value of spending today against the opportunity cost of having depleted funds tomorrow.
The budget constraint's introduction of coupling across rounds raises a sequence of challenging, fundamental questions: What is the baseline objective value an optimizer is guaranteed to obtain against any learner in a budgeted setting? What strategy achieves this? 
Are there strategically robust learning algorithms guaranteeing that the optimizer cannot extract anything beyond this baseline?

\paragraph{Our Contributions}
In this paper, we resolve these questions by characterizing the limits of manipulation in budgeted auction settings and identifying a simple and practical learning algorithm that achieves non-manipulability in second-price auctions, which, under certain assumptions, extends to first-price auctions.

\itparagraph{The Budgeted Stackelberg Equilibrium}
Our first contribution generalizes the classic notion of a Stackelberg equilibrium to accommodate cross-round constraints, introducing what we call the \emph{Budgeted Stackelberg Equilibrium} (BSE), \cref{def:defs:budgeted_stackelberg}.
In a standard repeated game, an optimizer achieves their Stackelberg value by repeatedly sampling actions from a single, optimal mixed strategy, and the learner learning to best respond to this.
When the optimizer has a budget constraint, we show that this approach is insufficient.
Because the optimizer's value as a leader is not necessarily concave as a function of their spending (see \cref{fig:defs:SE_BSE}), they can achieve strictly higher utility by modulating their spend rate over time---for instance, aggressively overspending in an initial phase of length $q T$, and underspending for the remaining $(1-q)T$ rounds.
We show that an optimizer with a $k$-dimensional budget constraint might need to use up to $k+1$ different strategies in $k+1$ distinct temporal phases (\cref{lem:defs:budgeted_stackelberg_simple,lem:defs:budgeted_stackelberg_simple_mult_constraints}).
For $k=0$, i.e., for an unbudgeted optimizer, this recovers the classic definition of Stackelberg equilibrium. 

Learning algorithms aim to best respond to the opponent's strategy. Even in the absence of budget constraints, employing a no-regret algorithm is insufficient for the learner to react to a time-multiplexed optimizer strategy, and we would need no-adaptive regret algorithms.
But from a \emph{budget-constrained learner's} perspective, a time-multiplexing optimizer would make it impossible for them to best respond, even against the single best distribution of their actions in hindsight.~\cite{ISSS22} shows this is indeed the case, as the learner cannot avoid linear regret in an adversarial bandits with knapsacks setting\footnote{In fact, the adversary in the counterexample from~\cite{ISSS22} looks similar to the optimizer's two-phase strategy that we use.}.
Without knowledge of when and what the optimizer would do in the future, a budget-constrained learner's reasonable strategy is to optimize their value while spending their budget evenly over time. This boils down to making the learner ``best respond'' to each of the optimizer's phases assuming the optimizer employs a stochastic strategy in each phase.
Existing learning algorithms for budgeted auctions, such as \cite{BG19}, do learn to optimize value in each phase and switch fast enough, even when the occurrence of the change of phase is unclear.
This allows a strategic optimizer to achieve their BSE value. We discuss this further in~\cref{remark:defs:opt_achievable,rem:auctions:BSE_achievable}, and show that the PID controller learning algorithm we analyze can achieve this best response in repeated auction settings.
The key question we ask next is if the optimizer can do better by changing their strategy at a higher speed and frequency (possibly using a super-constant number of changes), so that the learner will not be able to adjust.

Before proceeding to whether the optimizer can do better than their BSE, we remark that in a BSE, the optimizer's budget constraint forces them to $k+1$ phases (with a unique mixed strategy in each phase), even if the learner had no budget constraint, and would instantly best respond to each phase.
This is unlike an optimizer with no budget constraint, where, if the learner  instantly best-responds to each phase, the optimizer would find \emph{no additional utility in employing more than $1$ phase}.

\itparagraph{Strategic Robustness of simple PID Controller Algorithm}
While so far we have discussed the challenges in generalizing the Stackelberg equilibrium to set the lower bound for how much utility an optimizer should hope to get, the intriguing next step is whether the optimizer can extract $\Omega(T)$ more than what they can in the BSE.
For instance, can the optimizer exploit the learner's slow responses by using a super-constant ($\omega(1)$) number of phases or slowly drifting their strategy, rather than sticking to just $2$ phases as in their BSE? 
Indeed, while fast convergence on stationary phases suffices against constantly many switches, frequent switching or per-round adaptation can keep the learner permanently in transition, which could lead to manipulation, just as mean-based and gradient learners are manipulated for $\Omega(T)$ extra utility in unbudgeted games despite fast stationary convergence~\cite{BMSW18, DSS19, MMSS22} (see \cref{sec:app:OGD_manipulable} for a self-contained example with Online Gradient Descent, an analogous PID-style controller).

Our main contribution is addressing this intriguing question of non-manipulability in budgeted auction settings.
We consider the simple and well-studied bidding algorithm referred to as pacing~\cite{BG19, DBLP:journals/ior/ConitzerKSM22, DBLP:journals/mansci/ConitzerKPSMSW22}, where a player bids $\l$ times their value, for some value $\l$ and iteratively updates $\l$.
With the right choice of $\l$, this strategy is optimal in second-price auctions, and more generally in any single-parameter truthful auction~\cite{BG19}. 
Our learning algorithm, much like a PID controller, adjusts this pacing multiplier $\l$ proportionally to the over/under spending of each round, which converges to the optimal such $\l$ against stationary environments.
While this use of pacing multipliers is optimal only in truthful auctions, they are the go-to choice for bidding under budget constraints in very large-scale ad platforms in industry~\cite{GooglePacing, TwitterPacing, DBLP:journals/mansci/ConitzerKPSMSW22}, even when the auction is non-truthful.
Due to this, even though our main focus is second-price, we also study the above algorithm when the underlying auction is first-price. We show that our results extend to first-price auctions too, under the assumption that the learner's action space is restricted to exclusively using pacing strategies (this restriction is w.l.o.g. for second-price auctions).

Our main result (\cref{thm:main}) is that this PID controller learning algorithm is \textit{robust to strategic manipulation} when the optimizer and learner values are repeatedly drawn from a bounded joint distribution (but independent across rounds).
Specifically, we show that an optimizer cannot gain value more than their BSE value + $\order*{T^{2/3}}$ (\cref{thm:main}).
We show this under a mild distributional assumption\footnote{We require that $\mathbb{P} [0 < v_{\text{learner}} < \epsilon \cdot v_{\text{optimizer}}] = 0$ for some small $\epsilon$. On the other hand, when this is not true, the optimizer can obtain $\Omega(T^{1-\e})$ more than their BSE value by exploiting the transition periods of the learner (\cref{sec:app:dynamic:example}).}, which is satisfied when the smallest non-zero point in the distribution's support is bounded away from zero (e.g., values being in pennies). 
Our result establishes that the PID controller is not merely a practical heuristic for pacing in auctions, but, despite its predictability due to the deterministic updates, a strategically robust and non-manipulable learning algorithm.
It formally bridges the gap between the widespread use of the PID controller for bidding in large corporations~\cite{GooglePacing, TwitterPacing, DBLP:journals/mansci/ConitzerKPSMSW22}, and the rigorous non-manipulability guarantees called for in the strategic environments they get deployed in.

\paragraph{Technical Overview}
Our results require bridging techniques from online learning, control theory, and the geometry of repeated games. The technical core of the paper is divided into two parts: establishing the structural properties of the BSE (in \cref{sec:def} for general games and in \cref{sec:auctions} for auctions), and, the main technical part of our paper, proving the non-manipulability of the PID controller learning algorithm (\cref{sec:warmup,sec:dynamic}).

\itparagraph{The Geometry of Budgeted Stackelberg Equilibria}
To maximize their total return, a budgeted optimizer must operate on the concave closure of their utility-spending curve. Geometrically, the feasible region of the optimizer's utility and payments is constrained by the learner's best responses, making it non-convex. By allowing the optimizer to time-multiplex, this feasible region expands to the convex hull of these points (see \cref{fig:defs:feasible_region}). We show that by taking the convex combination of two points on the original utility curve, the optimizer can achieve strictly more utility than with any single fixed distribution. We generalize this structural insight to establish that under $k$ different budget constraints, the optimal strategy requires at most $k+1$ distributions, leveraging Carathéodory's theorem (\cref{lem:defs:budgeted_stackelberg_simple,lem:defs:budgeted_stackelberg_simple_mult_constraints}).

\itparagraph{Bounding Manipulability in Auctions via Lagrangian Relaxation}
The primary technical challenge lies in proving that an optimizer cannot exploit the predictable, deterministic nature of this PID update rule.
Since the learner might play a different pacing multiplier in every round, we cannot use classic notions of swap-regret that are the standard way to prove non-manipulability (see~\cite{DSS19,MMSS22}).
Since budgeted learning algorithms fail to satisfy classical learning properties (given our earlier discussion about no-regret being impossible), we propose a novel approach and show that, while the pacing multiplier varies across rounds deterministically, the optimizer cannot take advantage of such transitions. 
To this end, we formulate the adaptive optimizer's problem as a continuous-state dynamic program over the learner's pacing state $\l$ and analyze the recursive maximum Lagrangian reward $R_\tau(\lambda)$ the optimizer can achieve when there are $\tau$ rounds remaining and the learner is currently at pacing multiplier $\lambda$. 
We prove by induction that this reward is bounded by an affine form $R_\tau(\lambda) \le A \tau + \frac{1}{\eta}G(\lambda)$, where $A$ acts as the average per-round reward and $G(\lambda)$ is a state potential that bounds the optimizer's benefit from the learner occupying state $\lambda$. The crux of the proof relies on carefully constructing $G(\lambda)$ so that its derivative aligns with a Lagrangian dual variable $\gs(\lambda)$ corresponding to the learner's budget constraint for a fixed $\l$. However, a direct application fails because the optimal Lagrange dual variable might be infinite for certain $\l$. 
Relaxing our goal of using the optimal dual variable to simply coming up with a dual variable that is ``good enough'' still leads to problems: $\gs(\lambda)$ can be highly discontinuous, invalidating standard Taylor approximations.
To remedy this, we introduce a novel smoothing technique, leveraging the monotonicity of $\gs(\lambda)$ (which follows directly from primal optimality of the BSE benchmark) and averaging $\gs(\lambda)$ over a small interval $\sigma$ to produce a strictly Lipschitz continuous surrogate $\gs_\sigma(\lambda)$.
By integrating this smoothed dual variable, we calculate an appropriate $G(\l)$ and successfully bound the optimizer's benefit from the drift of the PID controller, proving that the learner adapts fast enough to cap the optimizer's total extracted value to the BSE baseline plus an additive error of $\order*{T^{2/3}}$.
To build intuition, we first present a simplified proof in \cref{sec:warmup} assuming uniform value distributions among bidders, followed by the formal arguments in \cref{sec:dynamic}.

\subsection*{\phantomsection Related Work}\label{sec:related}

In this section, we summarize some key past work that is related to our paper, focusing on the recent line of work on strategizing against learners.
We include a more extended presentation in \cref{sec:app:related}, where we focus on online learning with constraints and the usage of pacing multipliers.

\itparagraph{Strategizing Against Learners} A recent line of work has pursued the study of utility-maximizing responses to learning algorithms, and has asked what such a utility maximizer, referred to as the optimizer, earns in the long run, and what learning algorithms are not strategically manipulable. Initiated by the work of~\citet{BMSW18} in the setting of auctions, this has blossomed into a thriving area with such questions being asked in various settings including auctions~\cite{BMSW18,DSS19a,RZ24,cai2026onlinelinearoptimizationsufficient,DBLP:journals/corr/abs-2601-03853},  normal-form games~\cite{DSS19,CHJ20,BSV23,HPY24}, Bayesian games~\cite{MMSS22, RZ24}, polytope games~\cite{ACMMSS25}, principal-agent problems~\cite{GKSTVWW24}. Some important insights that come out of this line of work include the fact that learners using a no-swap regret algorithm will cap the optimizer's utility (up to $o(T)$ terms) at the Stackelberg value of the single-shot game~\cite{DSS19}, and that the optimizer can always earn the Stackelberg value.~\cite{MMSS22} made this connection tight by showing that no-swap regret is also necessary: i.e., there exists games where the optimizer can exceed Stackelberg value by exactly a constant multiple of the learner's swap regret. In the more general class of polytope games, introduced in~\cite{MMSS22}, that includes Bayesian games and extensive form games as special cases,~\cite{ACMMSS25} showed the right generalization of swap regret to ensure non-manipulability is the notion of profile swap regret.

\itparagraph{Learning and Bidding in Budget-Constrained Auctions}
A rich line of previous work has focused on no-regret guarantees for learning in budgeted auctions.
\cite{BG19, FPW23, AFZ25, WYDK23} learn optimal bidding strategies in first- and second-price auctions with budget and ROI constraints.
\cite{DBLP:conf/innovations/GaitondeLLLS23, DBLP:conf/colt/LucierPSZ24, DBLP:journals/mor/FikiorisT25} focus on learning optimal pacing multipliers and the implied welfare guarantees.
We refer the reader to our extended related work in \cref{sec:app:related} and the survey by \cite{AggarwalBBBDFGLLMMMMLPP24}.

\itparagraph{Learner vs Optimizer with Budgets}
To the best of our knowledge, the overlap between the two lines of work, namely, strategizing against learners and bidding under budget constraints, is limited to two prior works.~\cite{fikioris2026robust} study the same learning algorithm as us, but in a very restricted setting where all the players' values are $1$.
Their main result is that, when all players use this algorithm, each player's number of wins is approximately proportional to their budget share, even in short intervals.
They also show that in the simple setting they consider, all players running this algorithm is an equilibrium: no player can gain significantly more wins by a different strategy.
Due to the simplicity of their value distribution of all $1$s, the BSE in their setting is devoid of all the complexities we have to deal with in our paper: there is no need for multiple action distributions, and the BSE is the same as the Second Price Pacing Equilibrium (which can be thought of as the Nash equilibrium).
In a posted-pricing setting, \citet{GJLM23} study an unconstrained seller learning a static price via episodic binary search against a budget- and ROI-constrained buyer; they achieve no-regret against the optimal posted price (i.e., against the Stackelberg Equilibrium) when the buyer achieves low regret in each phase where the seller uses a fixed price.
While \citet{GJLM23} focus on how the optimizer can learn to achieve the Stackelberg Equilibrium value, our focus is what learning algorithms are robust enough to ensure that the optimizer cannot get more.

\paragraph{Subsequent Work on General Budgeted Games} Subsequent to posting this paper, we also considered a budgeted learner and optimizer in two-player normal-form games \cite{FST26general}.
We consider the two papers as related but distinct; here we focus on the foundations of BSE and on budgeted auctions, the application of learning in games with budget constraints that is most important in practice, where the space of best responses simplifies to a simple pacing multiplier due to the structure of the game.
We show that a \emph{simple, practical, and widely-used} algorithm for Bayesian budgeted auctions is strategy proof.
Our analysis is based on a form of dynamic programming based on the scalar state of the learner.
In contrast, in~\cite{FST26general} we offer a much more complex, strategy proof learning algorithm for general bi-matrix games based on calibrated forecasts with low error in an $L^2([-1, 1])$ space.
This complicated approach is unavoidable since in general games the learner has no information about the structure of the optimizer's utility. 

In the presence of~\cite{FST26general} it is natural to ask if one can cast the Bayesian auction as a normal-form game and use our solution in~\cite{FST26general} even for the auction setting. The answer is no, because, this requires treating entire bidding functions as pure strategies, which makes the number of actions for the learner infinite, or exponential if discretized, which also makes the bounds and running time in~\cite{FST26general} exponential. 

In this sense, our work in this paper and our subsequent work in~\cite{FST26general} answer complementary questions---is the algorithm used in practice robust (this paper)? and is robustness achievable in every
budgeted game \cite{FST26general}?---and neither result implies the other. The
controller in this paper has no counterpart in general games, where no deterministic learner is robust. The
same pattern is familiar without budgets: against a revenue-maximizing seller, simple no-regret bidding
algorithms (in the right parameterization) are already robust \cite{cai2026onlinelinearoptimizationsufficient, KSS24, DBLP:journals/corr/abs-2601-03853}, whereas in general
games, robustness requires low swap regret \cite{MMSS22}. As often in algorithm design, the special case
admits a strikingly simple algorithm, while the general case requires more complex algorithms.

\section{Budgeted Stackelberg Equilibrium} \label{sec:def}

In this section, we formally define the Budgeted Stackelberg Equilibrium (BSE), along with what it means for a learning algorithm to be manipulable in general games (see Section~\ref{sec:auctions} for the auctions setting).
Before doing so, we present a simple example to illustrate the benefit of the optimizer using multiple action distributions rather than using just one like in the classical Stackelberg Equilibrium.

\begin{wrapfigure}{R}{.32\textwidth}
    \centering
    \includegraphics[width=.32\textwidth]{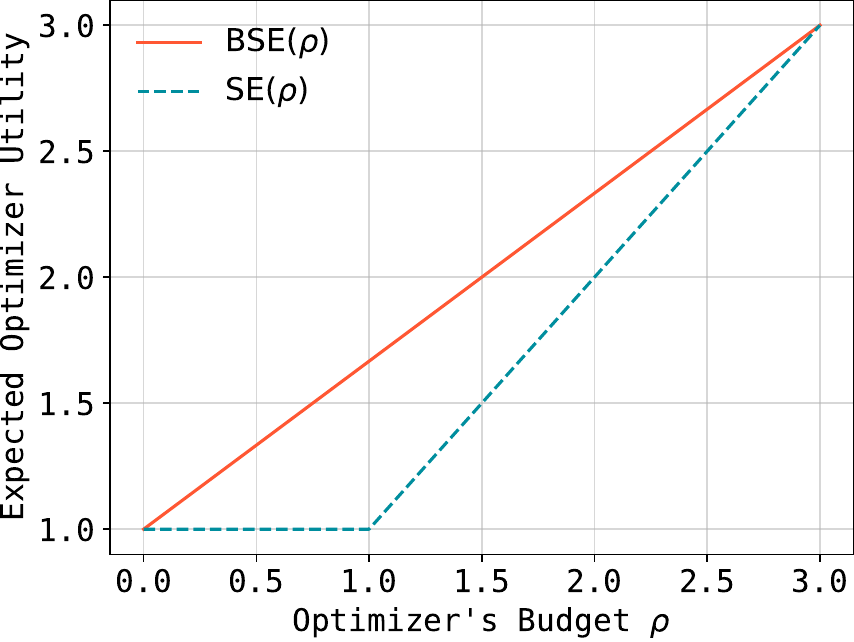}
    \caption{Optimal optimizer value by a single strategy obeying the budget (SE) and two strategies (BSE).}
    \label{fig:defs:SE_BSE}
\end{wrapfigure}

\paragraph{Illustrative Example}
Here, we present a simple example of how using multiple action distributions can be strictly stronger than using only one distribution, even when the learner best responds immediately.
Switching strategies is beneficial in unbudgeted games only if the learning algorithm is slow at best responding. We define the following game between two players, the optimizer (leader) and the learner (follower), each having two actions.
For simplicity, we assume that only the optimizer has a budget constraint, unlike our main results in \cref{sec:dynamic}.
To define the game, consider the following matrices:
\begin{equation*}
    \UO
    =
    \m{3&0\\0&1}
    \qquad
    \UL
    =
    \m{3&0\\0&1}
    \qquad
    \PO
    =
    \m{3&0\\0&0}
\end{equation*}

where\footnote{We denote with $\Delta(\calA)$ the distribution over some set $\calA$. For simple cases we abuse this notation: e.g., $\Delta([k])$ is the set of $k$-dimensional nonnegative vectors whose elements sum up to $1$.}, if $x \in \Delta([2])$ and $y \in \Delta([2])$ are the action distributions of the optimizer and the learner, respectively, then $x^\top \UO y$ is the optimizer's expected utility, $x^\top \UL y$ is the learner's expected utility, and $x^\top \PO y$ is the optimizer's expected payment, which she wants to keep less than $\rho$, for some $\rho \in [0, 3]$. 
In the classic Stackelberg Equilibrium (where the optimizer has no budget constraint), we pick $x$ and $y$ so that $y$ is a best response to $x$, while maximizing the optimizer's utility.
Extending this idea, we simply add a constraint for the optimizer's budget; formally, the optimizer's value when her budget share is $\rho$ is $\textrm{SE}(\rho)$:
\begin{equation} \label[optimization]{opt:defs:SE}
    \textrm{SE}(\rho) = \quad
    \max_{x, y} \quad x^\top \UO y
    \quad
    \textrm{ such that } \quad x^\top \PO y \le \rho
    \quad \textrm{ and } \quad
    y \in \argmax_{y'} x^\top \UL y'
\end{equation}

The blue dashed line in \cref{fig:defs:SE_BSE} shows the optimal value $SE(\rho)$.

Now consider the repeated setting, where the optimizer and the learner are playing the above game $T$ times, and the optimizer has a total budget of $T \rho$ with $\rho = 1/2$.
According to $SE(1/2)$, the best strategy is to always play her second action, which, assuming that the learner is approximately best responding, would result in her picking her second action in $T - o(T)$ rounds, resulting in $T - o(T)$ total utility for the optimizer.
However, the optimizer can do better by switching her action distribution once, assuming the learner will best respond appropriately.
Playing her first action $q T$ times, and then the second $(1-q)T$ times, the optimizer pays at most $3qT$, so with an average budget of $\rho = 1/2$, she can afford $q=\rho/3=1/6$.
If the learner best responds in $T - o(T)$ rounds, this leads to $\frac{4}{3}T - o(T)$ utility.
We can think of this as the optimizer mixing the two strategies corresponding to $SE(0)$ and $SE(3)$.
In fact, the optimal value of the BSE for any $\rho$ is achieved by mixing these two (for the right $q$), and its value is the concave closure of the values achieved by $\textrm{SE}(\rho)$, see \cref{fig:defs:SE_BSE}.
This increase in utility comes from the BSE allowing the convex hull of optimizer utility/payment pairs allowed by SE, which can be a strictly larger set, see \cref{fig:defs:feasible_region}.

\paragraph{Formal Definition}
Next, we give our formal definition for the \textit{Budgeted Stackelberg Equilibrium} (BSE) notion that we illustrated in the previous example.
For simplicity in this section, we give the definition when the optimizer (leader) has only one budget constraint, since this will be the focus of the rest of the paper.
We define it for an optimizer with multiple constraints in \cref{sec:app:defs}.

\begin{wrapfigure}{R}{.32\textwidth}
    \centering
    \includegraphics[width=.32\textwidth]{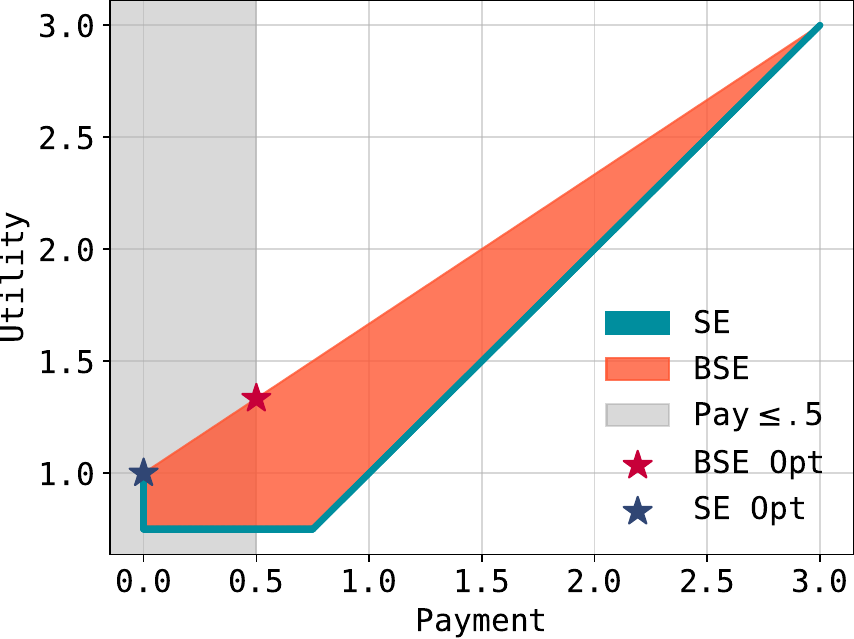}
    \caption{Possible optimizer utility and payment pairs when the learner best responds (feasible set of SE) and their convex hull (feasible set of BSE), along with example optimizer's budget restriction.}
    \label{fig:defs:feasible_region}
\end{wrapfigure}

Our definition abstracts away the learner's behavior by a best response function: for an action distribution $A$ of the optimizer, $\BR(A)$ are the learner action distributions that are best responses.
This $\BR$ function is simple to define for a learner who does not have any constraints (like in our previous example), but it is hard to universally define for a learner with a budget constraint like the one we consider later.
The standard approach for learning algorithms with budget constraints is to maximize the expected reward while trying to spend the budget evenly across time, which is what we also propose in \cref{sec:auctions}. This is optimal in Bayesian environments and the standard approach for adversarial environments \cite{BCCF24, FT23}.

\begin{definition}[Budgeted Stackelberg Equilibrium] \label{def:defs:budgeted_stackelberg}
    Consider strategy spaces of two players, $\calA$ for the optimizer (leader) and $\calB$ for the learner (follower).
    For actions $a \in \calA$ and $b \in \calB$ let $U(a, b) \in \R$ be the utility of the optimizer and $P(a, b) \in \R$ be the payment of the optimizer for a resource with average budget $\rho \in \R$.
    For an optimizer's action distribution $A \in \Delta(\calA)$, let $BR(A) \sub \Delta(\calB)$ be the learner action distributions that are best responses to $A$.
    Then, a \textit{Budgeted Stackelberg Equilibrium} (BSE) is a distribution $\nu$ over $\Delta(\calA) \times \Delta(\calB)$ (i.e., a distribution over independent distributions) that maximizes the expected utility of the optimizer while satisfying the budget constraint in expectation and the learner is best responding in every pair in the support of $\nu$.
    Formally, the value of the BSE is
    \begin{optimization} \label{opt:defs:budgeted_stackelberg}
        \sup_{\nu \in \Delta\qty\big( \Delta(\calA) \times \Delta(\calB) )}
        \quad &
        \Ex[(A, B) \sim \nu]{
            \Ex[a \sim A, b\sim B]{U(a, b)}
        }
        \\
        \textrm{\textup{ such that }} \quad
        \quad &
        \Ex[(A, B) \sim \nu]{
            \Ex[a \sim A, b\sim B]{P(a, b)}
        } \le \rho
        \qquad \textrm{\textup{ and }} \qquad
        B \in \BR(A) \quad \forall (A, B) \in \textrm{supp}(\nu)
    \end{optimization}
\end{definition}

 We now simplify the above definition.
Specifically, we show that the optimal distribution $\nu$ can be supported on at most $2$ points.
More generally, in \cref{sec:app:defs}, we show that when the optimizer has $k$ constraints, the optimal $\nu$ can be supported on at most $k + 1$ points.

\begin{lemma} 
\label{lem:defs:budgeted_stackelberg_simple}
    The optimum of \cref{opt:defs:budgeted_stackelberg} is obtained using a $\nu$ supported on two points only.
\end{lemma}

Let $U(A, B) = \Ex[a \sim A, b \sim B]{U(a, b)}$ and similarly for $P(A, B)$.
We prove the lemma by considering all the points $\qty\big( U(A, B), P(A, B) )$ we can generate, under the condition $B \in \BR(A)$.
The optimal distribution $\nu$ of \cref{opt:defs:budgeted_stackelberg} is a convex combination of these points (see the orange region in \cref{fig:defs:feasible_region}).
Since these points are on a $2$-dimensional space, by Carath\'eodory's theorem, any such distribution $\nu$ can be supported on $3$ points.
However, the optimal solution $\nu$ is on the boundary of the feasible region, since the objective is linear w.r.t. $\nu$, which is actually a convex combination of at most $2$ points.

We now formally state our observation from our previous example, that allowing randomization over action distributions can lead to strictly more utility for the optimizer.

\begin{theorem}
    The value of a BSE can be strictly more than the value of a Stackelberg equilibrium (the latter restricts \cref{opt:defs:budgeted_stackelberg} to a single $(A, B) \in \Delta(\calA) \times \Delta(\calB)$ pair).
\end{theorem}

\begin{remark} \label{remark:defs:opt_achievable}
    A reasonable question is whether the optimizer can always achieve utility close to $\opt \cdot T$, where $\opt$ is her BSE value.
    For unbudgeted players \cite{DSS19} prove that, if the learning algorithm guarantees $o(T)$ regret, under mild conditions on the game, the optimizer can get $\opt \cdot T - o(T)$ utility, guiding the learner to the desired best response. 
    If the optimizer alone has budget constraints, we can still think of a learner who uses classical algorithms.
    Since the optimizer might need to change her action distribution, she needs the learner to adapt to this change.
    This means that classic no-regret algorithms may fail to switch their actions fast enough. However, slightly stronger properties guarantee fast enough responses, e.g., if the learning algorithm has $o(T)$ \emph{adaptive} regret \cite[Section 10.2]{Hazan16}. 
    In budgeted auction settings, we identify the corresponding property that our learning algorithm satisfies in~\cref{rem:auctions:BSE_achievable}.
\end{remark}

\paragraph{Manipulability of a learning algorithm}
Here, we formalize our definition of what it means for a learning algorithm to be manipulable.
\cite{DSS19, MMSS22, ACS24, ACMMSS25} define manipulability in games without budgets, and consider any $\Omega(T)$ utility increase due to time-varying strategies to be manipulation.
In the absence of budgets, this is equivalent to the optimizer getting $\Omega(T)$ more than her Stackelberg utility.
However, a budgeted optimizer may want to switch strategies due to budget constraints and not to gain utility while the learner is adapting.
Due to this, we directly define the manipulability of a learning algorithm, based on the comparison between the maximum utility the optimizer can get and her BSE.

\goodbreak
\begin{definition}[Manipulability] \label{def:defs:manipulability}
    Consider a game as defined in \cref{def:defs:budgeted_stackelberg} and let $\opt$ be the value of its Budgeted Stackelberg Equilibrium.
    A learning algorithm is called \textit{non-manipulable} if every strategy for the optimizer that satisfies her budget constraints achieves expected utility at most $\opt \cdot T + o(T)$.
\end{definition}

\section{Repeated Budgeted Auctions} \label{sec:auctions}

We consider two budget-limited players, the \textit{learner} and the \textit{optimizer}, participating in $T$ repeated auctions.
In every round $t$, the learner observes a value $\vL^{(t)} \in [0, 1]$ and the optimizer observes a value $\vO^{(t)} \in [0, 1]$, drawn from a joint distribution $\qty\big( \vL^{(t)}, \vO^{(t)} ) \sim \mathcal D$.
Each then submits a bid; the highest bid wins and makes a payment; we consider either first-price auctions (where the winner pays her bid) or second-price auctions (where the winner pays the loser's bid).
The learner's total budget is $T \rL$, and the optimizer's is $T \rO$, for some $\rL, \rO > 0$.
Each player's goal is to maximize the total value they get over the $T$ rounds, subject to their total payment being less than their budget with probability $1$.

\paragraph{Learning Algorithm for Auctions}
We consider a simple learning algorithm for the learner's bids.
Specifically, in round $t$ she uses a pacing multiplier $\l^{(t)} \ge 0$ and bids $\l^{(t)}\vL^{(t)}$.
She then observes some payment $p_{\mathtt{L}}^{(t)}$ and updates her pacing multiplier using the following PID controller
\begin{equation} \label{eq:update}
    \l^{(t+1)}
    =
    \l^{(t)}
    +
    \eta \qty\big( \rL - p_{\mathtt{L}}^{(t)} )
    .
\end{equation}
In other words, the learner increases/decreases her pacing multiplier proportional to her under/over payment compared to her per-round budget $\rL$.
This algorithm is similar to the one proposed by \cite{fikioris2026robust}, who used it in a setting similar to ours but where $\vL = \vO = 1$ in every round.
While simple, we will show that this algorithm is resistant to strategic manipulation.

Pacing multipliers have been widely studied in equilibria for budgeted auctions or as simple bidding strategies to learn in online environments.
More importantly for us, if the optimizer were to use any fixed bidding distribution that does not respond to the learner's bidding, pacing multipliers are optimal for second-price auctions, and our learning algorithm will find such a pacing multiplier. However, \cref{eq:update} aims to find a good pacing multiplier regardless of the auction format.
While using a pacing multiplier is not optimal for first-price auctions, many works try to find the optimal such multiplier and show their benefits for welfare guarantees \cite{DBLP:conf/innovations/GaitondeLLLS23, DBLP:journals/mor/FikiorisT25, DBLP:conf/colt/LucierPSZ24, DBLP:journals/mansci/ConitzerKPSMSW22}.
\cite{BKK23} also examine equilibria in budgeted first-price auctions and show that first pacing one's values and then composing it with a symmetric unbudgeted Bayesian Nash Equilibrium (BNE) strategy results in an equilibrium.

\cref{eq:update} never produces negative pacing multipliers or runs out of budget if initialized with $\lambda^{(1)} \le \rL$:

\begin{restatable}[Similar to {\cite[Lemma 2.1]{fikioris2026robust}}]{lemma}{AlgorithmProps} \label{lem:auction:learner_budget}
    Assume the learner is bidding using a pacing multiplier, updated by \cref{eq:update}, with $\eta \in (0, 1)$ and $\l^{(1)} \ge 0$.
    Assume the auction format is either first- or second-price.
    Then, for arbitrary behavior by the optimizer, with probability $1$ it holds that $\l^{(t)} \ge 0$ and the learner's total payment by round $T$ is $\rL T + \frac{\l^{(1)} - \l^{(T+1)}}{\eta}$.
    In addition, if $\l^{(1)} \le \rL$, then the learner never runs out of budget.
\end{restatable} 

The total payment rule follows by simple manipulation of \cref{eq:update}.
We prove that if $\l^{(t)} \le \rL$, then the pacing multiplier cannot decrease in the next round, proving that $\l^{(1)} \le \rL$ leads to no budget violation.

\paragraph{Auction definitions and optimizer strategies}
Since the learning algorithm is deterministic, the pacing multiplier used by the learner is always predictable.
Therefore, instead of using bids as the optimizer's action space, we interpret her bids as using the learner's pacing multiplier, scaled by a ``fake'' value.
Specifically, her action space each round is a function $\hatvO: [0, 1] \to \Rz$ that maps her value $\vO$ to some ``fake'' value $\hatvO(\vO)$.
Then, when the learner is using pacing multiplier $\l$ (and therefore bids $\l\,\vL$), the optimizer is bidding $\l\,\hatvO(\vO)$.
This simplifies who wins (we break ties in favor of the learner): if $\hatvO(\vO) > \vL$ the optimizer wins, otherwise the learner wins.
Then, when the optimizer uses $\hatvO$ and the learner $\l$, with values $(\vL, \vO)$, we denote with $\UO(\hatvO, \vO, \vL)$ and $\l \PO(\hatvO, \vO, \vL)$ the optimizer's utility and payment; we denote with $\l \PL(\hatvO, \vO, \vL)$ the learner's payment.
For example, for second-price auctions, we have $\UO(\hatvO, \vO, \vL) = \vO \One{\hatvO(\vO) > \vL}$, $\l \PO(\hatvO, \vO, \vL) = \l \vL \One{\hatvO(\vO) > \vL}$, and $\l \PL(\hatvO, \vO, \vL) = \l \hatvO(\vO) \One{\hatvO(\vO) \le \vL}$.
For first-price the payments become $\l \PO(\hatvO, \vO, \vL) = \l \hatvO(\vO) \One{\hatvO(\vO) > \vL}$ and $\l \PL(\hatvO, \vO, \vL) = \l \vL \One{\hatvO(\vO) \le \vL}$.
This makes the learner's update in round $t$ when the optimizer uses $\hatvO^{(t)}$ to be $\l^{(t+1)} = \l^{(t)} + \eta \qty\big( \rL - \l^{(t)} \PL(\hatvO^{(t)}, \vO^{(t)}, \vL^{(t)}) )$.
We also overload these functions; for a fake value function $\hatvO \in \Rz^{[0, 1]}$ and a randomized fake value function $\hatVO \in \Delta\qty\big(\Rz^{[0, 1]})$, we define (and similarly define $\PO(\hatvO)$, $\PO(\hatVO)$, $\PL(\hatvO)$, $\PL(\hatVO)$)
\begin{equation*}
    \UO(\hatvO) = \Ex[(\vO, \vL) \sim \mathcal D]{\UO(\hatvO, \vO, \vL)}
    \quad\text{and}\quad
    \UO(\hatVO) = \Ex[\hatvO \sim \hatVO]{\UO(\hatvO)} = \Ex[\hatvO \sim \hatVO]{ \Ex[(\vO, \vL) \sim \mathcal D]{\UO(\hatvO, \vO, \vL)} }
    .
\end{equation*}

\paragraph{Budgeted Stackelberg Equilibrium in Budgeted Auctions}
We now specialize the BSE (\cref{def:defs:budgeted_stackelberg}), simplifying it using \cref{lem:defs:budgeted_stackelberg_simple}, in the context of repeated budgeted auctions, using the functions defined in the previous part.
Here, given a bidding function $\hatvO \in \Rz^{[0, 1]}$ and the learner's multiplier $\l$, the optimizer's expected utility is $\UO(\hatvO)$, her expected payment is $\l \PO(\hatvO)$, and her per-round budget is $\rO$.
In a second-price auction, a fixed pacing multiplier is the optimal bidding with respect to any static bidding by the optimizer.
In general, the higher the pacing multiplier, the more the learner's value and payment, making the optimal pacing multiplier to be $\rL / \PL(\hatvO)$.
To simplify our future results, we relax this constraint and define as $\BR(\hatvO)$ any $\l$ such that $\l \PL(\hatvO) \ge \rL$; we explain next why this does not change the value of the BSE, $\opt$, in this setting.
Formally,
\begin{optimization} \label{opt:auctions:stackelberg_eq}
    \opt = \;
    \sup_{\substack{
        \hatVO^{(1)}, \hatVO^{(2)} \in \Delta\qty\big(\Rz^{[0, 1]}),\\
        \l_1, \l_2 \in \Rz, \, 0 \le q \le 1
    }}
    \quad &
    q \, \UO\qty\big(\hatVO^{(1)}) + (1-q) \UO\qty\big(\hatVO^{(2)})
    \\
    \textrm{such that }
    \qquad &
    q \,\l_1 \PO\qty\big(\hatVO^{(1)}) + (1-q) \l_2 \PO\qty\big(\hatVO^{(2)}) \le \rO
    \quad\textrm{ and }\quad
    \lambda_i \PL\qty\big(\hatVO^{(i)}) \ge \rL
    \;\; \forall i \in [2]
    .
\end{optimization}

The reason we use an inequality rather than an equality for the learner's payment is that, for any feasible solution, we can decrease a $\l_i$ until the learner's inequality is an equality without affecting feasibility or increasing the objective value.
This simplifies some of our future results.

\begin{remark}[Optimizer's ability to get $\opt \cdot T$] \label{rem:auctions:BSE_achievable}
    We now comment on whether the optimizer can always achieve value $\opt \cdot T$ in repeated budgeted auctions.
    Consider any (near) optimal solution to \cref{opt:auctions:stackelberg_eq} with pacing multipliers $\l_1 \le \l_2$.
    In rounds $t \le T q$, the optimizer can bid $\l^{(t)} \hatvO(\vO^{(t)})$ where $\hatvO \sim \hatVO^{(1)}$; this leads to expected utility of $\UO(\hatVO^{(1)})$ for those rounds, as suggested by the equilibrium.
    In expectation, this leads to the learner's update to be 
    $
        \l^{(t+1)} - \l^{(t)}
        =
        \eta\qty\big( \rL - \l^{(t)} \PL(\hatVO^{(1)}) )
        \le
        \eta \rL \qty\big( 1 - \frac{\l^{(t)}}{\l_1} )    
    $, 
    i.e., the learner's pacing multiplier performs a random walk biased towards something less than $\l_1$.
    Therefore, it will converge within $\order*{1/\eta} = o(T)$ steps to $\l_1$.
    We can similarly examine the second phase, the last $T (1-q)$ rounds.
    This implies that the learner converges to the desired multipliers within $o(T)$ rounds, leading to at most $o(T)$ overspending by the optimizer, implying an $\opt \cdot T - o(T)$ utility for the optimizer.
    Note, however, that this simple convergence argument only applies when the optimizer switches strategies a constant number of times; if an adaptive optimizer switches strategies every $\Theta(1/\eta)$ rounds or varies $\hatVO^{(t)}$ continuously with $\l^{(t)}$, the learner spends the entire horizon in transition without converging as we illustrate for Online Gradient Descent in \cref{sec:app:OGD_manipulable}, which one can also think of as a PID-style controller. The example shows that showing the non-manipulability of our controller requires the analysis of \cref{sec:warmup,sec:dynamic}.
\end{remark}

\paragraph{Lagrangian Relaxation}
It will be more convenient to relax the optimizer's budget constraint using a Lagrange multiplier $\mu \ge 0$.
Specifically, we show that strong duality holds with respect to this constraint.

\begin{restatable}{lemma}{StrongDuality} \label{lem:auction:strong_duality}
    Assume that both the optimizer and the learner have positive expected values.
    Then, strong duality holds: there exists\footnote{Formally, this optimal $\mu$ is the one that minimizes $R^\star(\mu) + \mu \rO$.} a $\mu \ge 0$ such that $\opt = R^\star + \mu \rO$, where $R^\star$ is the maximum Lagrangian reward the optimizer can achieve for this optimal $\mu$:
    \begin{optimization} \label{opt:auctions:stackelberg_eq_lagrange}
        R^\star = \quad
        \sup_{
            \hatVO \in \Delta(\Rz^{[0, 1]}), \,
            \l \in \Rz
        }
        \quad
        \UO\qty\big(\hatVO) - \mu \l \PO\qty\big(\hatVO)
        \qquad
        \textrm{\textup{ such that }}
        \qquad
        \lambda \PL\qty\big(\hatVO) \ge \rL
        .
    \end{optimization}
\end{restatable}

We prove this lemma via standard convex optimization arguments: we analyze the $2$-dimensional region defined by all the possible expected utility and payments of the optimizer, constrained by the learner's best response.
Given that our final solution is a point from the convex hull of that region, this makes \cref{opt:auctions:stackelberg_eq} a convex program, and we establish strong duality via Slater's condition by finding a point in the interior of the feasible region.
We present the full proof in \cref{sec:app:auctions}.

\cref{lem:auction:strong_duality} shows that, instead of focusing on the total value accumulated by the budgeted optimizer, we can consider the total Lagrangian reward they receive.
Specifically, we show that the total Lagrangian reward is in expectation at most $R^\star T + \order*{T^{2/3}}$, which immediately proves our main theorem:

\begin{theorem} \label{thm:main}
    Fix any joint distribution of values such that $\Pr{0 < \vL < \vO\e} = 0$ for some $\e > 0$.
    Assume that the optimizer and learner have budgets $\rO T$ and $\rL T$, and they participate in $T$ repeated first- or second-price auctions.
    Let $\opt$ be the value of the BSE (\cref{opt:auctions:stackelberg_eq}), when the learner's best response function is the optimal pacing multiplier for the auction format being used.
    If the Learner bids using pacing multipliers $\l^{(t)}$, updated using \cref{eq:update} with step size $\eta \in (0, 1)$ and initial value $\l^{(1)} \le \rL$, then the Optimizer's expected total value is at most\footnote{Our $\order*{\cdot}$ terms hide terms relating to the players' average budgets or their value distributions.} $\opt \cdot T + \mathcal O\qty\big(T \sqrt \eta + \frac{1}{\eta})$.
    If $\eta = \Theta(T^{-2/3})$, this becomes $\opt \cdot T + \order*{T^{2/3}}$.
\end{theorem}

We note that the above distributional assumption is necessary.
There are value distributions where the optimizer can extract $\opt \cdot T + \Theta(T^{1 - \e})$ for any constant $\e > 0$ if the learner has significant probability mass near 0.
We include the details of this example in \cref{sec:app:dynamic:example}.

To provide some intuition, before presenting the proof for arbitrary budget shares and distributions, we consider in the next section the example where both players have uniform distributions (even though this is not covered by our distributional assumption).

\section{Non-manipulability Warm-up Example -- Uniform Distributions} \label{sec:warmup}
In this section, we consider a simple example where the optimizer cannot get substantially more value than her BSE value.
We examine budgeted second-price auctions where both players' values $\vO$ and $\vL$ are independently drawn from the uniform $[0, 1]$ distribution and their total budgets are $T$, i.e., $\rO = \rL = 1$. 
We start by showing the optimizer's value at the Nash and Budgeted Stackelberg equilibria, and how the second one is larger.
Then we outline the proof that the PID controller learning algorithm of \cref{eq:update} is not manipulable in this case. 
In \cref{sec:dynamic}, we extend the argument to prove the general case of \cref{thm:main}.

\paragraph{Equilibria of the two bidder auction game}
To contrast with the proposed solutions of past works, we first consider the static solution of the \textit{Second Price Pacing Equilibrium} (SPPE) \cite{DBLP:journals/ior/ConitzerKSM22}\footnote{In the SPPE definition of \cite{DBLP:journals/ior/ConitzerKSM22}, players are not allowed to bid more than their value. We do not place this restriction, i.e., there is no ROI constraint. For this example, we can recover the original definition of SPPE by scaling the values up.}, where both players use pacing multipliers to bid.
By symmetry, they use the same pacing multiplier $\l$, which gets them expected value $\Ex{\vO \One{\l \vO \ge \l \vL}} = 1/3$ and they are paying $\Ex{\l\vL \One{\l \vO \ge \l \vL}} = \l / 6$, so need to set $\l = 6$.
If the learner uses $\l = 6$ to bid, doing the same is the best response for the optimizer.

\begin{wrapfigure}[15]{R}{.32\textwidth}
    \centering
    \includegraphics[width=\linewidth]{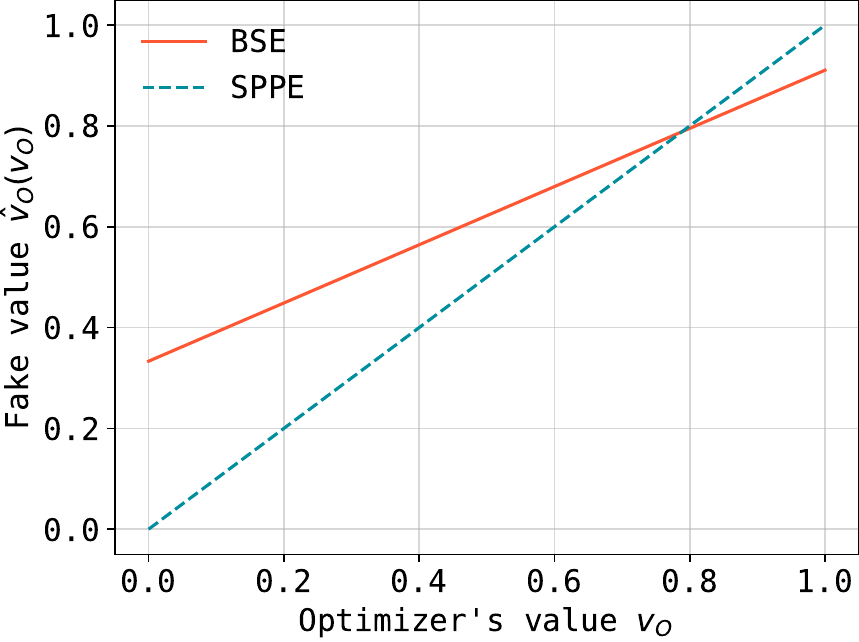}
    \caption{Optimal ``fake'' value $\hatvO(\cdot)$ (i.e., bidding $\l \hatvO(\vO)$ when the learner uses $\l$) for the SPPE (leading to $0.33$ optimizer value) and the BSE (leading to $0.36$ optimizer value).}
    \label{fig:warmup:BSE_SPPE}
\end{wrapfigure}

Now consider the BSE, where the optimizer does not have to best respond.
We can show in this simplified setting that the optimizer does not have to randomize over two action distributions (note that this is not generally true; in \cref{sec:app:multiple_dists_example} we show an example where the optimizer randomizes over two distributions of actions).
In fact, her optimal bid (see \cref{fig:warmup:BSE_SPPE}) is deterministic given her value, specifically, $\hatvO^\star(\vO) = \nicefrac{\vO}{\sqrt 3} + \nicefrac{1}{3}$ and $\l^\star = \nicefrac{18}{2 + \sqrt 3} \approx 4.82$, with expected optimizer value $\opt = \nicefrac{3 + 2 \sqrt 3}{18} \approx 0.36$, which is more than $1/3$ as in the SPPE.

Note that the optimizer's bid when her value is small seems abnormally high, e.g., $\hatvO^\star(0) = \frac{1}{3}$, winning with probability $1/3$.
This is due to the optimizer having a secondary objective to keep the learner's payment high (second inequality in \eqref{opt:auctions:stackelberg_eq}), incentivizing high bidding even with low values.
We also note that this is not a Bayes-Nash equilibrium: if the learner commits to bidding $\l^\star \vL$, there is a pacing multiplier by the optimizer that yields more than $0.36$ utility.

\paragraph{Non-manipulability of the PID controller}
We now consider the manipulability of our learning algorithm of \cref{eq:update}.
\textit{Can the optimizer achieve expected value $\opt \cdot T + \Omega(T)$?}
We answer this negatively for this example, giving intuition on how we prove this for general distributions in \cref{sec:dynamic} (under mild distributional assumptions).

We consider the Lagrangian relaxation of the optimizer's budget constraint using a Lagrange multiplier $\mu$ as we did in \cref{lem:auction:strong_duality}.
We can show that the optimal value for $\mu$ is $\nicefrac{3 + 2\sqrt 3}{54}$, making $R^\star = \nicefrac{3 + 2 \sqrt 3}{27}$.
We now upper-bound the total Lagrangian reward the optimizer can get by $R^\star T + o(T)$, implying the desired bound for the budget-constrained value.
The benefit of working with the Lagrangian reward is that there is a simple recursive formula that calculates its optimal expected reward.
Let $R_\tau(\l)$ be the maximum expected Lagrangian reward the optimizer can get if there are $\tau$ rounds left, and the learner is currently using pacing multiplier $\l$.
$R_0(\l) = 0$ and for $\tau \ge 1$,
\begin{equation} \label{eq:warmup:R_def}
    R_\tau(\l) = 
    \max_{\hatvO:[0,1] \to [0, 1]}\,
    \Ex[\vL,\vO]{\bigg.
        (\vO - \mu \l \vL) \One{\hatvO(\vO) \ge \vL}
        +
        R_{\tau - 1}\qty\Big(
            \l + \eta\qty\big( 1 - \l \hatvO(\vO) \One{\hatvO(\vO) < \vL} )
        )
    }
\end{equation}
i.e., the optimizer maximizes her current Lagrangian reward plus whatever reward she gets in the subsequent $\tau - 1$ rounds, given the learner's update for $\l$.
We want to upper bound the total Lagrangian reward the optimizer gets when there are $T$ rounds left, and the learner starts at $0$, i.e., $R_T(0)$.
We do this by the following bound for every $\tau, \l$:
\begin{equation} \label{eq:warmup:separation_bound}
    R_\tau(\l)
    \le
    A \, \tau + \frac{1}{\eta} G(\l)
\end{equation}
for some $A \ge 0$ and function $G: \Rz \to \R$.
The value $A$ represents the average reward per round, and $G(\l)$ is proportional to the benefit of the learner being at a certain pacing multiplier $\l$.
We notice that if $G(\l) \ge 0$, then \cref{eq:warmup:separation_bound} holds for $\tau = 0$.
We will show \cref{eq:warmup:separation_bound} for $\tau \ge 1$ inductively; assume that \cref{eq:warmup:separation_bound} holds for $\tau - 1$.
Then, by \cref{eq:warmup:R_def}, we have that
\begin{equation*}
    R_\tau(\l) \le
    A (\tau - 1)
    +
    \max_{\hatvO:[0,1] \to [0, 1]}
    \Ex[\vL,\vO]{\bigg.
        (\vO - \mu \l \vL) \One{\hatvO(\vO) \ge \vL}
        +
        \frac{1}{\eta} G\qty\Big(
            \l + \eta\qty\big( 1 - \l \hatvO(\vO) \One{\hatvO(\vO) < \vL} )
        )
    }
\end{equation*}

To show \eqref{eq:warmup:separation_bound} for this $\tau$ and all $\l$, we can prove that the above r.h.s. is at most $A \tau + \frac{1}{\eta}G(\l)$, i.e., show
\begin{align} \label{eq:warmup:A_bound}
    \forall \l \ge 0: \;\;
    A \; &\ge
    \max_{\hatvO:[0,1] \to [0, 1]}
    \Ex[\vL,\vO]{\bigg.
        (\vO - \mu \l \vL) \One{\hatvO(\vO) \ge \vL}
        +
        \frac{1}{\eta}G\qty\Big(
            \l + \eta\qty\big( 1 - \l \hatvO(\vO) \One{\hatvO(\vO) < \vL} )
        )
        - \frac{1}{\eta}G(\l)
    }
    \nonumber\\
    &\approx
    \max_{\hatvO:[0,1] \to [0, 1]}
    \Ex[\vL,\vO]{\bigg.
        (\vO - \mu \l \vL) \One{\hatvO(\vO) \ge \vL}
        +
        G'(\l) \qty\big( 1 - \l \hatvO(\vO) \One{\hatvO(\vO) < \vL} )
    }
    +
    \order{\l^2 \eta}
\end{align}
where $G'(\cdot)$ is the derivative of $G(\cdot)$ and in the last approximation we take the second order Taylor expansion $G(\l + \e) - G(\l) \approx \e G'(\l) + \order*{\e^2}$, assuming it is accurate.
If we show that the maximum in the r.h.s. of the above inequality is at most $R^\star$, we would get the desired claim.
To that end, we define for every $\l, g$
\begin{align*}
    f(\l, g)
    &=
    \max_{\hatvO:[0,1] \to [0, 1]}
    \Ex[\vL,\vO]{\bigg.
        (\vO - \mu \l \vL) \One{\hatvO(\vO) \ge \vL}
        +
        g \qty\big( 1 - \l \hatvO(\vO) \One{\hatvO(\vO) < \vL} )
    }
    \\&=
    \max_{\hatvO:[0,1] \to [0, 1]}\int_0^1\qty(
        \vO \, \hatvO(\vO)
        -
        \mu \, \l\,  \frac{\hatvO(\vO)^2}{2}
        +
        g\cdot \qty\Big( 1 - \l\, \hatvO(\vO) \qty\big( 1 - \hatvO(\vO) ) )
    ) d\vO
\end{align*}
where in the last equality we used that $\vL,\vO$ are sampled from the uniform distribution.
We note that the above maximum is easy to compute for $g \le 0$, since in that case the function inside the integral is concave w.r.t. $\hatvO(\vO)$, making the unconstrained optimum $\hatvO(\vO) = \frac{\vO - \l \, g}{\l (\mu - 2 g)}$ (explaining the affine form of the optimal solution for this special case).
Then, for every $\l$, we want to pick $g$ to get $A \approx R^\star$.
There are multiple ways to approach this.
A natural first idea is to pick $g$ to minimize $f(\l, g)$, since it serves as a lower bound for $A$ (\cref{eq:warmup:A_bound}).
However, while this serves as a good bound for $A$, it invalidates our previous assumptions and goals.
Specifically, for $\l < 4$ where $\l\, \hatvO(\vO) \qty\big( 1 - \hatvO(\vO) ) < 1$ for all $\hatvO$, we have that $\argmin_g f(\l, g) = -\infty$, making it unclear how to define $G(\cdot)$ for that range.

Since for some $\l$ it holds that $f(\l, -\infty) = -\infty$, a different idea, instead of minimizing $f$, is to pick $g$ so that for every $\l$, $f(\l, g) = R^\star$; that still gets $A \approx R^\star$ for all $\l$.
However, this equation might have multiple solutions, which is impossible to handle for general distributions.
In addition, even in this example, some solutions invalidate our previous assumptions.
Specifically, one such solution satisfies $g = -\Omega(\l^{-1/2})$ for large $\l$ (see $g_{\textrm{bad}}$ in \cref{fig:warmup:f}), implying that its integral would be $G(\l) = \Omega(\sqrt\l)$, making it too large and unnatural: as $\l$ grows larger, the optimizer's payment increases, implying $G(\l)$ should decrease.

\begin{wrapfigure}{R}{.32\textwidth}
    \centering
    \includegraphics[width=\linewidth]{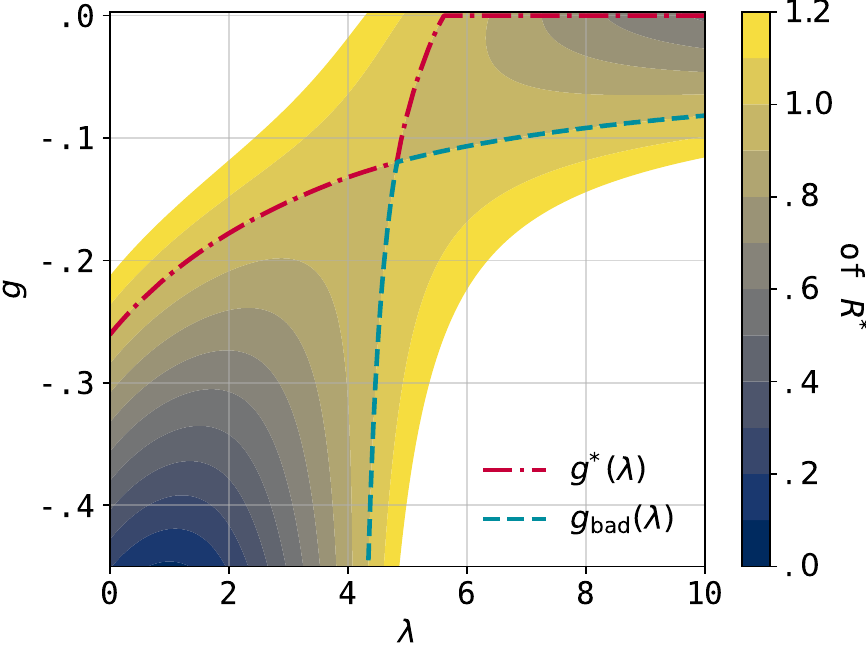}
    \caption{$f(\l, g)$ relative to $R^\star$, along with bad solution of the $f(\l, g) = R^\star$ equation ($g_{\textrm{bad}}$) and our final solution $\gs(\l)$, see \cref{eq:warmup:g_star}.}
    \label{fig:warmup:f}
\end{wrapfigure}

We pick $g$ according to $\gs(\l)$, which we defined inspired by the following three points.
First, to get $A \approx R^\star$ from \cref{eq:warmup:A_bound}, all we need is $f(\l, \gs(\l)) \le R^\star$ for all $\l \ge 0$.
Second, to get $G(\l) = \int_\l \gs(x) dx$ to be as small as possible for all $\l$, we should try to make $\gs(\l)$ as close to $0$ as possible.
Third, given that we expect $G(\l)$ to be decreasing, it should be $\gs(\l) \le 0$ for all $\l$.
These lead to the following definition:
\begin{equation} \label{eq:warmup:g_star}
    \gs(\l) = \sup\qty{\big.
        g \le 0 : f(\l, g) \le R^\star
    }
\end{equation}

Specifically, \cref{fig:warmup:f} shows how $f(\l, g)$ behaves relative to $R^\star$.
The same figure also shows pairs $(\l, g)$ that satisfy $f(\l, g) \le R^\star$, along with the resulting $\gs(\l)$.
We also observe some other key properties of $\gs(\l)$, that we prove in \cref{ssec:dynamic:gs_well_behaved} for general value distributions:
\begin{itemize}
    \item $\gs(\l)$ is bounded for all $\l$, i.e., there exists $g \le 0$ with $f(\l, g) \le R^\star$.
    \item $\gs(\l)$ is weakly increasing, making it integrable.
    \item $\gs(\l)$ eventually becomes $0$, implying its integral is always bounded.
\end{itemize}

Another property that holds for uniform distributions but not in general is Lipschitzness of $\gs$, which implies that the Taylor approximation that we made in \cref{eq:warmup:A_bound} is accurate.
All these properties imply that using $G(\l) = -\int_\l^\infty \gs(x) dx$ gives us the desired bound of $A = R^\star + \order{\bar\l^2\eta}$, where $\bar \l$ is such that $\gs(\bar \l) = 0$.
This yields that the optimizer's expected Lagrangian reward is $R_T(0) \le R^\star T + \order*{T \eta + \frac{1}{\eta}}$.

To extend these ideas to general value distributions in \cref{sec:dynamic}, we need to handle one more issue, that $\gs(\l)$ might not be Lipschitz continuous; in fact, it can be discontinuous.
We solve this by considering a smoothed version of $\gs$, $\gs_\sigma(\l) = \frac{1}{\sigma}\int_\l^{\l + \sigma} \gs(x) dx$, which (by boundedness) is $\order{1/\sigma}$-Lipschitz and satisfies $f(\l, \gs_\sigma(\l)) \le R^\star + \order{\sigma}$.
Then, using $G(\l) = -\int_\l^\infty \gs_\sigma(x) dx$ gives the following bound $R_T(0) \le R^\star T + \order*{T \sigma + \frac{\eta}{\sigma} T + \frac{1}{\eta}}$, which, if optimized over $\sigma, \eta$, implies $R_T(0) \le R^\star T + \order*{T^{2/3}}$.

\section{Non-manipulability of the Budgeted Learner} \label{sec:dynamic}

In this section, we prove our main theorem, \cref{thm:main}, that our proposed learning algorithm is non-manipulable, i.e., that the optimizer cannot gain substantially more than her BSE value (\cref{opt:auctions:stackelberg_eq}) when the learner is responding with her optimal pacing multiplier (which is optimal among all bidding strategies for second-price auctions).
Specifically, if the learner bids $\l^{(t)} \vL^{(t)}$ in round $t$, and uses \cref{eq:update} to update $\l^{(t)}$, the optimizer cannot gain more than $\opt \cdot T + o(T)$ expected value.
Our theorem holds for both first- and second-price auctions with the following distributional assumption: there exists some $\e > 0$ such that $\Pr{0 < \vL < \vO\e} = 0$, which is satisfied if, e.g., possible values are multiples of a small constant $\delta$.
We note that this distributional assumption is necessary: in \cref{sec:app:dynamic:example} we show an example where $\Pr{0 < \vL < \vO\e} = \Omega(\e^\d)$ for small $\d > 0$ and the optimizer can get $\opt \cdot T + \Omega\qty\big(T^{1 - \order*{\d}})$ value in expectation.

To prove \cref{thm:main}, we follow the following steps, similar to our warm-up example in \cref{sec:warmup}.
First, we relax the optimizer's budget constraint and upper bound her total Lagrangian reward, similar to \cref{opt:auctions:stackelberg_eq_lagrange}.
Then, we upper bound this Lagrangian reward by $A \, \tau + \nicefrac{1}{\eta} G(\l)$, when there are $\tau$ rounds remaining and the learner is currently using pacing multiplier $\l$.
By achieving $A = R^\star + o(1)$ and $G(0) = \order*{1}$ (where $R^\star$ is the maximum Lagrangian reward of the Optimizer's problem, \cref{opt:auctions:stackelberg_eq_lagrange}), we get the theorem.
To pick the right $G(\l)$ function that bounds the learner's benefit for a particular $\l$, we examine the dual function $f(\l, g)$ after also Lagrangifying the learner's constraint.
Unlike \cref{sec:warmup}, here we cannot calculate in closed form the above quantities; in fact, the definition of the dual function $f(\l, g)$ is a complicated maximization program over all the bidding functions that the optimizer can use.
In \cref{sec:warmup} we showed by examining \cref{fig:warmup:f} that the $\gs(\l)$ (that subsequently defines $G(\l)$) satisfies key properties: monotonicity, boundedness, eventually becoming $0$, and Lipschitz continuity.
We also show that the $\gs(\l)$ for general value distributions satisfies monotonicity (\cref{lem:dynamic:g_monotone}) and boundedness (\cref{lem:dynamic:g_at_zero}).
Unfortunately, there are distributions (like the one in \cref{sec:app:dynamic:example}) where $\gs(\l) \ne 0$ for all $\l \ge 0$; we show that this is not the case when the value distribution satisfies the aforementioned property (or that $R^\star$ is above a certain threshold, see \cref{lem:dynamic:g_becomes_zero}).
Finally, we have to deal with the issue that the function $\gs(\cdot)$ can be discontinuous.
We resolve this final issue by taking an averaged version of $\gs(\cdot)$ on an interval of small length $\sigma$ that we show is Lipschitz continuous and approximately satisfies all the properties of $\gs(\cdot)$.

\subsection{Relaxing the Upper Bound for the Optimizer's Problem}

By leveraging strong duality (\cref{lem:auction:strong_duality}), and in particular \cref{opt:auctions:stackelberg_eq_lagrange}, instead of upper bounding the total value the optimizer can get, we upper bound her maximum Lagrangian reward.
Formally, the maximum value the optimizer can get when the learner starts at some $\l^{(1)}$, given she wants to observe her budget constraint with probability $1$, is
\begin{optimization*}
    \sup_{\hatvO^{(t)} \in \Rz^{[0, 1]}}
    \quad &
    \Ex{\sum_{t=1}^T \UO\qty(\hatvO^{(t)}, \vL^{(t)}, \vO^{(t)})}
    \\
    \textrm{such that } \quad &
    \sum_{t=1}^T \l^{(t)} \PO\qty(\hatvO^{(t)}, \vL^{(t)}, \vO^{(t)}) \le \rO T
    \\ &
    \l^{(t + 1)} = \l^{(t)} + \eta\qty(\bigg. \rL - \l^{(t)} \PL\qty(\hatvO^{(t)}, \vL^{(t)}, \vO^{(t)}) )
    \qquad \forall t \in [T]
\end{optimization*}
where the inequality must hold for every realization of $(\vL^{(t)}, \vO^{(t)}) \sim \mathcal D$ and each $\hatvO^{(t)}$ depends only on the values of the previous rounds (and not future ones).
We will assume that the initial value $\l^{(1)}$ is small enough so that the learner does not run out of budget (\cref{lem:auction:learner_budget}).
By relaxing the optimizer's budget constraint to hold only in expectation, lagrangifying that constraint with the multiplier $\mu \ge 0$ of \cref{lem:auction:strong_duality}, we get an upper bound for the above problem:
\begin{optimization*}
    \sup_{\hatvO^{(t)} \in \Rz^{[0, 1]}}
    \quad &
    \Ex{\sum_{t=1}^T \qty( \bigg.
        \UO\qty(\hatvO^{(t)}, \vL^{(t)}, \vO^{(t)})
        -
        \mu \, \l^{(t)} \PO\qty(\hatvO^{(t)}, \vL^{(t)}, \vO^{(t)})
        + \mu \rO
    )}
    \\
    \textrm{where } \quad &
    \l^{(t + 1)}
    =
    \l^{(t)} + \eta\qty(\bigg. \rL - \l^{(t)} \PL\qty(\hatvO^{(t)}, \vL^{(t)}, \vO^{(t)}) )
    \qquad\qquad \forall t \in [T]
    .
\end{optimization*}

Ignoring the constant $\mu \rO$ term in the objective, we can write the above recursively.
In particular, it is equal to $R_T\qty\big(\l^{(1)})$, where
\begin{equation*}
    R_\tau(\l) =
    \begin{cases}
        \displaystyle
        \sup_{\hatvO \in \Rz^{[0, 1]}}\;\Ex[\vL, \vO]{\bigg.
            \UO\qty(\hatvO, \vL, \vO)
            -
            \mu \l \PO\qty(\hatvO, \vL, \vO)
            +
            R_{\tau-1}\qty\Big(
                \l + \eta\qty\big( \rL - \l \PL\qty(\hatvO, \vL, \vO) )
            )
        }
        &\;\textrm{ if } \tau \ge 1
        \\
        0&\;\textrm{ o/w}
    \end{cases}
    ,
\end{equation*}
i.e., $R_\tau(\l)$ denotes the total Lagrangian reward the optimizer can get if there are $\tau$ rounds remaining and the learner is currently using pacing multiplier $\l$.
If we were to show that $R_T(\l) \le R^\star \cdot T + o(T)$, then by strong duality we would get \cref{thm:main}.
However, proving such a bound directly is complicated.
In fact, the $\hatvO$ that is the maximizer of the above expression depends on both $\tau$ and $\l$, making it impossible to compute.
Instead, we relax $R_\tau(\l)$ even further by separating the two variables, via a function of the form $A \tau + \frac{1}{\eta}G(\l)$.
Specifically, we prove the following lemma that gives such a bound.

\begin{lemma} \label{lem:dynamic:separation}
    For every $\tau \ge 0$ and $\l \ge 0$, the Lagrangian reward of the optimizer is upper-bounded by
    \begin{equation}\label{eq:dynamic:separation_bound}
        R_\tau(\l)
        \le
        A \tau + \frac{1}{\eta} G(\l)
    \end{equation}
    as long as the following two conditions hold
    \begin{enumerate}
        \item $G(\l) \ge 0$ for all $\l \ge 0$.
        \item For all $\l \ge 0$
        \begin{equation} \label{eq:dynamic:prop2}
            A \ge
            \sup_{\hatvO \in \Rz^{[0, 1]}}\qty(
                \UO\qty(\hatvO)
                -
                \mu \l \PO\qty(\hatvO)
                +
                \frac{1}{\eta}\Ex[\vL, \vO]{\bigg.
                    G\qty\Big(
                        \l + \eta\qty\big( \rL - \l \PL\qty(\hatvO, \vL, \vO) )
                    ) - G(\l)
                }
            )
        \end{equation}
    \end{enumerate}
\end{lemma}

\begin{proof}
    We prove \cref{eq:dynamic:separation_bound} inductively in $\tau$.
    First, for $\tau = 0$, the claim holds since $G(\l) \ge 0 = R_0(\l)$ for all $\l \ge 0$.
    Fix a $\tau \ge 1$, and assume \cref{eq:dynamic:separation_bound} holds for $\tau-1$.
    We have that
    \begin{alignat*}{3}
        \Line{
            R_\tau(\l)
        }{=}{
            \sup_{\hatvO \in \Rz^{[0, 1]}}\qty(
                \UO\qty(\hatvO)
                -
                \mu \l \PO\qty(\hatvO)
                +
                \Ex[\vL, \vO]{\bigg.
                    R_{\tau-1}\qty\Big(
                        \l + \eta\qty\big( \rL - \l \PL\qty(\hatvO, \vL, \vO) )
                    )
                }
            )
        }{}
        \\
        \Line{}{\le}{
            \sup_{\hatvO \in \Rz^{[0, 1]}}\qty(
                \UO\qty(\hatvO)
                -
                \mu \l \PO\qty(\hatvO)
                +
                A(\tau-1)
                +
                \frac{1}{\eta}\Ex[\vL, \vO]{\bigg.
                    G\qty\Big(
                        \l + \eta\qty\big( \rL - \l \PL\qty(\hatvO, \vL, \vO) )
                    )
                }
            )
        }{}
        \\
        \Line{}{\le}{
            A(\tau-1)
            +
            A + \frac{1}{\eta} G(\l)
            =
            A \tau + \frac{1}{\eta} G(\l)
        }{}
    \end{alignat*}
    where the first inequality follows from the induction hypothesis (\cref{eq:dynamic:separation_bound} holds for $\tau-1$); the second follows from \cref{eq:dynamic:prop2} and careful rearranging.
\end{proof}

\subsection{Finding an appropriate \texorpdfstring{$G$}{G} for Lemma \ref{lem:dynamic:separation}} \label{ssec:dynamic:find_G}

In this section, we find a function $G$ (and a value $A$) that satisfies \cref{lem:dynamic:separation} and, in particular, \cref{eq:dynamic:separation_bound}.
To draw intuition for how we set this $G$, we consider the following calculation.
Consider that $G$ is ``nice'' and that $G(\l + \e) \approx G(\l) + \e G'(\l)$, where $G'(\cdot)$ is the derivative of $G(\cdot)$.
In this case, the r.h.s. of \cref{eq:dynamic:prop2} approximately becomes
\begin{equation*}
    \sup_{\hatvO \in \Rz^{[0, 1]}}\qty(\Big.
        \UO\qty(\hatvO)
        -
        \mu \l \PO\qty(\hatvO)
        +
        G'(\l)\qty\big( \rL - \l \PL\qty(\hatvO) )
    )
\end{equation*}
and given that we want to pick the smallest possible $A$, it makes sense to set:
\begin{equation} \label{eq:dynamic:min_max_lagrangian}
    A =
    \inf_{G':\Rz \to \R}\,
    \sup_{\l \ge 0}\,
    \sup_{\hatvO \in \Rz^{[0, 1]}}\qty(\Big.
        \UO\qty(\hatvO)
        -
        \mu \l \PO\qty(\hatvO)
        +
        G'(\l)\qty\big( \rL - \l \PL\qty(\hatvO) )
    )
\end{equation}
which, since $G'(\l)$ is minimizing the expression for every $\l$, can also be rewritten as 
\begin{equation*}
    A =
    \sup_{\l \ge 0}\,
    \inf_{g \in \R}\,
    \sup_{\hatvO \in \Rz^{[0, 1]}}\qty(\Big.
        \UO\qty(\hatvO)
        -
        \mu \l \PO\qty(\hatvO)
        +
        g\cdot\qty\big( \rL - \l \PL\qty(\hatvO) )
    )
    .
\end{equation*}

We notice that the above problem is very similar to the dual of \cref{opt:auctions:stackelberg_eq_lagrange}.
In fact, the objective is equal to the Lagrangian of the following problem for a fixed $\l \ge 0$:
\begin{optimization} \label{opt:dynamic:opt_per_lambda}
    \max_{\hatVO \in \Delta(\Rz^{[0, 1]})}
    &\quad
    \UO\qty(\hatVO) - \mu \l \PO\qty(\hatVO)
    \\
    \textrm{such that }
    &\quad
    \rL = \l \PL\qty(\hatVO)
\end{optimization}

Therefore, assuming some form of strong duality, we might be able to set $G'(\lambda)$ to the optimal Lagrange multiplier $g$ for every $\l$ of the above problem to get a good $A$.
However, this poses major difficulties, as also outlined in \cref{sec:warmup}.
For example, for certain values of $\l$, the optimal value of $g$ might be $-\infty$, e.g., if $\l < \rL$ \cref{opt:dynamic:opt_per_lambda} is infeasible, since $\PL\qty\big(\hatVO) \le 1$.
While this makes sense from the setting of $A$ we have above, having $G'(\l) = -\infty$ does not correspond to anything useful for \cref{lem:dynamic:separation}.

Despite these problems, the above solution is a good starting point.
First, we define the dual function of the above problem, $f: \R^2 \to \R$, 
\begin{equation*}
    f(\l, g)
    =
    \sup_{\hatvO \in \Rz^{[0, 1]}}\qty(\Big.
        \UO\qty(\hatvO)
        -
        \mu \l \PO\qty(\hatvO)
        +
        g\cdot\qty\big( \rL - \l \PL\qty(\hatvO) )
    ).
\end{equation*}

We then pick $g$ so that $f(\l, g) \le R^\star$, with the goal to achieve $A \approx R^\star$ in \cref{lem:dynamic:separation}.
While this does not minimize the dual function for every $\l$ (like in \cref{eq:dynamic:min_max_lagrangian}), it achieves the goal of \cref{lem:dynamic:separation}.
Specifically, we define
\begin{equation} \label{eq:dynamic:g_def}
    \gs(\l)
    =
    \sup\qty{
        g \le 0: f(\l, g) \le R^\star
    }
    .
\end{equation}

Note that if for some $\l$, there is no $g \le 0$ such that $f(\l, g) \le R^\star$, then $\gs(\l) = -\infty$ (we prove this is not the case later).
However, if for some finite $g \le 0$ it holds that $f(\l, g) \le R^\star$, then we can replace the above $\sup$ with a $\max$, i.e., $f\qty\big(\l, \gs(\l)) \le R^\star$.
This is because the set $\{ g \le 0: f(\l, g) \le R^\star \}$ is an interval of the form $[a, b]$ or $(-\infty, b]$ for $a \le b \le 0$, due to $f(\l, \cdot)$ being convex.

\begin{claim} \label{cl:dynamic:f_convex}
    As a supremum of functions linear in $g$ or $\l$, $f(\l, g)$ is convex in each argument (but not necessarily jointly convex).
\end{claim}

Then, to guarantee that $G(\cdot)$ is not negative, we set 
\begin{equation*}
    G(\l) \approx
    - \int_\l^\infty \gs(x) dx
\end{equation*}

Note that even if $\gs$ is finite, $\gs$ might not be integrable or the above might be $+\infty$, e.g., if $\gs(\l) = - \Omega(1/\l)$ for large $\l$; we prove these are non-issues later, under the distributional assumption of \cref{thm:main}.
We now make some observations about the definition of $\gs$ in \cref{eq:dynamic:g_def} and the subsequent definition of $G$.
First, the definition of $G$ matches our original goal to set its derivative to something relating to a Lagrange multiplier of \cref{opt:dynamic:opt_per_lambda}.
Second, we define $\gs(\l)$ to be non-positive so that $G(\l)$ is weakly decreasing in $\l$, i.e., our bound $A \tau + \frac{1}{\eta}G(\l)$ is weakly decreasing in $\l$.
Third, because $\gs(\l) \le 0$, we are guaranteed to have that $G(\l) \ge 0$, as needed in \cref{lem:dynamic:separation}.
Finally, to make sure that the value of $G(\l)$ is as small as possible, in \cref{eq:dynamic:g_def} we define $\gs$ to be as close to $0$ as possible under the required assumption on the value of $f$.

\subsubsection{Proving that \texorpdfstring{$\gs$}{g star} is well behaved} 
\label{ssec:dynamic:gs_well_behaved}

The key property of $\gs$ of \cref{eq:dynamic:g_def} is that it is \textit{weakly increasing}.
This subsequently guarantees multiple desirable properties: it is integrable, bounded if and only if $\gs(0)$ is bounded, and its integral is non-infinite if it becomes $0$ for any $\l$.
We now prove that $\gs$ is weakly increasing.

\begin{lemma}[Monotonicity] \label{lem:dynamic:g_monotone}
    The function $\gs(\l)$, as defined in \cref{eq:dynamic:g_def}, is weakly increasing.
\end{lemma}

We prove the lemma by noticing that, if the above condition were violated, then there would be a feasible bidding that achieves a Lagrangian reward higher than $R^\star$.

\begin{proof}
    Towards a contradiction, assume that for some $0 \le \l_1 < \l_2$ it holds that $\gs(\l_1) > \gs(\l_2)$.
    Note that this implies that $\gs(\l_1) > - \infty$.

    First consider $f\qty\big(\l_2, \gs(\l_1))$.
    It must be the case that $f\qty\big(\l_2, \gs(\l_1)) > R^\star$ since, either
    (a) $\gs(\l_2) = - \infty$, implying $f\qty\big(\l_2, g) > R^\star$ for all $g \le 0$, or
    (b) $\gs(\l_2) \ne - \infty$, implying $f\qty\big(\l_2, g) > R^\star$ for all $g > \gs(\l_2)$.
    Let $\hatvO$ be the maximizer of $f\qty\big(\l_2, \gs(\l_1))$ (recall that $f$ is defined as a supremum over a Lagrangian).
    If no such maximizer exists, we define $\hatvO$ to be such that its value is close enough to $f\qty\big(\l_2, \gs(\l_1))$, so that it holds
    \begin{equation} \label{eq:dynamic:1}
        \UO(\hatvO) - \mu \l_2 \PO(\hatvO) + \gs(\l_1)\qty\big( \rL - \l_2 \PL(\hatvO) )
        >
        R^\star
    \end{equation}

    By definition of $\gs(\l_1) > - \infty$, we have that
    \begin{align*}
        R^\star
        \ge
        f\qty\big(\l_1, \gs(\l_1))
        \ge
        \UO(\hatvO) - \mu \l_1 \PO(\hatvO) + \gs(\l_1)\qty\big( \rL - \l_1 \PL(\hatvO) )
    \end{align*}

    Combining the above with \cref{eq:dynamic:1} we get
    \begin{equation*}
        (\l_2 - \l_1) \qty( \mu \PO(\hatvO) + \gs(\l_1) \PL(\hatvO)  )
        <
        0
        \iff
        \mu \PO(\hatvO) + \gs(\l_1) \PL(\hatvO) < 0
    \end{equation*}
    where we used that $\l_2 > \l_1$.

    We now assume that $\rL > \l_2 \PL(\hatvO)$ and prove the lemma; we postpone the proof of $\rL > \l_2 \PL(\hatvO)$ for the end of the proof.
    We define $\l_3 = \rL / \PL(\hatvO)$; this is well defined since $\PL(\hatvO) > 0$ by our claim above.
    We note that $\l_3 > \l_2$, since $\rL > \l_2 \PL(\hatvO)$.
    We now revisit \cref{eq:dynamic:1}:
    \begin{alignat*}{3}
        \Line{
            R^\star
        }{<}{
            \UO(\hatvO) - \mu \l_2 \PO(\hatvO) + \gs(\l_1)\qty\big( \rL - \l_2 \PL(\hatvO) )
        }{}
        \\
        \Line{}{<}{
            \UO(\hatvO) - \mu \l_3 \PO(\hatvO) + \gs(\l_1)\qty\big( \rL - \l_3 \PL(\hatvO) )
        }{\mu \PO(\hatvO) + \gs(\l_1) \PL(\hatvO) < 0\\\text{and }\l_3 > \l_2}
        \\
        \Line{}{=}{
            \UO(\hatvO) - \mu \l_3 \PO(\hatvO)
        }{\l_3 = \rL / \PL(\hatvO)}
    \end{alignat*}

    However, the above inequality is a contradiction by the definition of $R^\star$: this $(\hatvO, \l_3)$ combination is feasible for \cref{opt:auctions:stackelberg_eq_lagrange} but achieves a strictly higher value than the maximum.

    To complete the proof of the lemma, we conclude with the proof that $\rL > \l_2 \PL(\hatvO)$.
    Assume towards a contradiction that $\rL \le \l_2 \PL(\hatvO)$.
    This makes the pair $(\hatvO, \l_2)$ feasible for \cref{opt:auctions:stackelberg_eq_lagrange}, which implies that $\UO\qty(\hatvO) - \mu \l_2 \PO\qty(\hatvO) \le R^\star$.
    Combining this with \cref{eq:dynamic:1}, we get
    \begin{equation*}
        \gs(\l_1)\qty\big( \rL - \l_2 \PL(\hatvO) )
        >
        0
    \end{equation*}
    which implies that $\rL - \l_2 \PL(\hatvO) < 0$.
    Also, for every $g \le \gs(\l_1)$:
    \begin{equation*}
        f(\l_2, g)
        \ge
        \UO(\hatvO) - \mu \l_2 \PO(\hatvO) + g\cdot\qty\big( \rL - \l_2 \PL(\hatvO) )
        \ge
        \UO(\hatvO) - \mu \l_2 \PO(\hatvO) + \gs(\l_1)\qty\big( \rL - \l_2 \PL(\hatvO) )
    \end{equation*}
    where in the inequality we use $\rL < \l_2 \PL(\hatvO)$.
    The above implies that
    \begin{equation*}
        \inf_{g \le \gs(\l_1)} f(\l_2, g)
        \ge
        \UO(\hatvO) - \mu \l_2 \PO(\hatvO) + \gs(\l_1)\qty\big( \rL - \l_2 \PL(\hatvO) )
        >
        R^\star
    \end{equation*}
    where the last inequality is \cref{eq:dynamic:1}.
    On the other hand, we know that $\gs(\l_1) > \gs(\l_2)$, implying $f(\l_2, g) > R^\star$ for $g \ge \gs(\l_2)$.
    This implies that
    \begin{equation*}
        \inf_{g \le 0} f(\l_2, g)
        >
        R^\star
    \end{equation*}

    We proceed to show that the above is a contradiction, along with our previous observation, $\rL < \l_2 \PL(\hatvO)$.
    Consider the following optimization program
    \begin{optimization} \label{opt:dynamic:4}
        \sup_{\hatVO \in \Delta(\Rz^{[0, 1]})}
        &\quad
        \UO\qty(\hatVO) - \mu \l_2 \PO\qty(\hatVO)
        \\
        \textrm{such that }
        &\quad
        \l_2 \PL\qty(\hatVO) \ge \rL
    \end{optimization}

    The above problem is a convex optimization problem, by considering the subset of $\R^2$,
    \begin{equation*}
        F = \qty{\bigg.
            \qty(\Big.
                \UO\qty(\hatvO') - \mu \l_2 \PO\qty(\hatvO'),
                \l_2 \PL\qty(\hatvO')
            ):
            \hatvO' \in \Rz^{[0, 1]}
        },
    \end{equation*}
    its closure $\bar F$ and its convex hull $\conv(\bar F)$.
    Specifically, the convex problem is
    \begin{optimization*}
        \max_{x \in \conv(\bar F)}
        &\quad
        x[1]
        \\
        \textrm{such that }
        &\quad
        x[2] \ge \rL
    \end{optimization*}

    We now observe that strong duality holds.
    By our previous observations that $\rL < \l_2 \PL(\hatvO)$, i.e., $\exists x \in F$ such that $\rL < x[2]$, we get that there exists a neighborhood around this point that strictly satisfies the constraint.
    The intersection of this neighborhood and $\conv(\bar F)$ is in the relative interior of $\conv(\bar F)$ and strictly satisfies the constraint, which makes Slater's condition hold, and therefore strong duality.
    We notice now that $f(\l_2, g)$ is the dual function of \cref{opt:dynamic:4}.
    Given that $\inf_{g \le 0} f(\l_2, g) > R^\star$ and strong duality, we get that the objective of \cref{opt:dynamic:4} is also strictly bigger than $R^\star$.
    However, $R^\star$ was defined as \cref{opt:dynamic:4} where we are also allowed to optimize over our choice of $\l$, implying it cannot have a higher objective.
    This completes the contradiction.
\end{proof}

We now proceed to prove that $\gs$ is bounded.
Since it is weakly increasing, we only have to prove that $\gs(0)$ is bounded, which we do by explicitly calculating its value.

\begin{lemma}[Boundedness] \label{lem:dynamic:g_at_zero}
    The function $\gs(\l)$, as defined in \cref{eq:dynamic:g_def}, satisfies $\gs(0) = - \frac{\Ex{\vO} - R^\star}{\rL} \le 0$.
\end{lemma}

\begin{proof}
    We notice that
    \begin{equation*}
        f(0, g)
        =
        \sup_{\hatvO \in \Rz^{[0, 1]}}\qty(\Big.
            \UO\qty(\hatvO)
            +
            g \rL
        )
        =
        \Ex{\vO} + g \rL
    \end{equation*}

    Given the constraint $\Ex{\vO} + g \rL \le R^\star$, the highest $g$ that satisfies this is $- \frac{\Ex{\vO} - R^\star}{\rL}$.
    We notice that this is non-positive due to \cref{opt:auctions:stackelberg_eq_lagrange}, since the optimization objective is
    \begin{equation*}
        \UO(\hatVO) - \l\mu\PO(\hatVO)
        \le
        \UO(\hatVO)
        \le
        \Ex{\vO}
        .
    \end{equation*}

    This completes the proof.
\end{proof}

We now prove that $\gs(\l) = 0$ for large enough $\l$, under our distribution assumption in \cref{thm:main}.
In addition, we prove the claim when $R^\star$ is larger than the optimizer's value when the learner has a value of $0$.
Combined with the above lemma, this proves that $\int_0^\infty \gs(\l)d\l$ is bounded.
We note that the following lemma might not be true without the assumptions, e.g., in our example in \cref{sec:app:dynamic:example}, see \cref{fig:app:dynamic:f}.

\begin{lemma} \label{lem:dynamic:g_becomes_zero}
    Let $\gs(\l)$ be as defined in \cref{eq:dynamic:g_def}.
    Assume that $\Ex{\vL} > 0$.
    Assume that $\Pr{0 < \vL < \vO\e} = 0$ for some $\e > 0$ or that $R^\star > \Ex{\vO \One{\vL = 0} }$.
    Then, there exists a finite $\bar \l$ such that $\gs(\bar\l) = 0$.
\end{lemma}

\begin{proof}
    We first prove that $R^\star \ge \Ex{\vO \One{\vL = 0} }$, relying just on the assumption that $\Ex{\vL} > 0$.
    Consider the optimizer bidding a small constant $\d > 0$, such that $\Pr{\vL \ge \d} > 0$ (if no such $\d$ exists then $\Ex{\vL} = 0$).
    Let $\l = \rL/\PL(\d)$, which is well defined since $\PL(\d) \ge \d \Pr{\vL \ge \d} > 0$.
    This means that this $\l$, along with bidding $\d$, is feasible for \cref{opt:auctions:stackelberg_eq_lagrange}.
    We now examine its objective value.
    For second-price,
    \begin{equation*}
        \UO(\d) - \mu \l \PO(\d)
        =
        \UO(\d) - \mu \rL \frac{\PO(\d)}{\PL(\d)}
        =
        \Ex{\vO \One{\vL < \d} } - \mu \rL \frac{\Ex{\vL \One{\vL < \d}}}{\d \Pr{\vL \ge \d}}
        \ge
        \Ex{\vO \One{\vL < \d} } - \mu \rL \frac{\d \Pr{0 < \vL < \d}}{\d \Pr{\vL \ge \d}}
    \end{equation*}
    and for first-price
    \begin{equation*}
        \UO(\d) - \mu \l \PO(\d)
        =
        \UO(\d) - \mu \rL \frac{\PO(\d)}{\PL(\d)}
        =
        \Ex{\vO \One{\vL < \d} } - \mu \rL \frac{\d \Pr{\vL < \d}}{\Ex{\vL \One{\vL \ge \d}}}
    \end{equation*}

    We notice that in both cases, the final r.h.s. converges to $\Ex{\vO \One{\vL = 0} }$ by taking $\d \to 0$.
    This proves that $R^\star \ge \Ex{\vO \One{\vL = 0} }$.

    We now proceed to prove that $\gs(\l)$ becomes $0$ eventually.
    We get that $\gs(\l) = 0$ if $f(\l, 0) \le R^\star$.
    This is easy to do in the case when $\mu = 0$: we can extend the above analysis to show that $R^\star = \Ex{\vO}$ and also show that $f(\l, 0) = \Ex{\vO}$, making the desired $\bar \l = 0$.
    For $\mu > 0$, we have that
    \begin{equation*}
        f(\l, 0)
        =
        \sup_{\hatvO}\qty\Big(
            \UO\qty(\hatvO)
            -
            \mu \l \PO\qty(\hatvO)
        )
        \le
        \sup_{\hatvO} \Ex{\qty(\vO - \mu \l \vL) \One{\hatvO(\vO) > \vL}}
        =
        \Ex{(\vO - \mu \l \vL)^+}
    \end{equation*}
    where in the inequality we used the payment rule of second-price that lower bounds the first-price one and in the final equality we used the ``truthful'' bidding rule that maximizes the expression, $\hatvO(\vO) = \frac{\vO}{\mu\l}$.
    We further upper bound the above quantity:
    \begin{equation*}
        \Ex{(\vO - \mu \l \vL)^+}
        =
        \Ex{\vO \One{\vL = 0} }
        +
        \Ex{(\vO - \mu \l \vL) \One{0 < \mu \l \vL < \vO}}
        \le
        \Ex{\vO \One{\vL = 0} }
        +
        \Pr{0 < \mu \l \vL < \vO}
    \end{equation*}
    where in the inequality we used $\vO - \mu \l \vL \le 1$.
    We now analyze the above for the two assumptions we made in the lemma.
    If $\Pr{0 < \vL < \vO\e} = 0$ for some $\e > 0$ then setting $\l = \bar\l = 1/(\mu\e)$ makes the above equal to $\Ex{\vO \One{\vL = 0} } \le R^\star$.
    If, instead, $R^\star > \Ex{\vO \One{\vL = 0} }$ then since the probability converges to $0$ as $\l \to \infty$, there must be some finite $\bar\l$ such that the expression becomes at most $R^\star$.
    In either case, this proves that $f(\bar\l, 0) \le R^\star$ for some finite $\bar\l$, implying $\gs(\bar\l) = 0$.
\end{proof}

We now examine the last property of $\gs$ we are interested in.
At the beginning of \cref{ssec:dynamic:find_G}, we analyzed the requirements of \cref{lem:dynamic:separation} under the assumption that $G(\l + \e) \approx G(\l) + \e G'(\l)$, which is similar to assuming that $G'(\l)$ is Lipschitz.
However, we can show that there are examples where $\gs(\l)$ is not only not Lipschitz, but also discontinuous, implying that the previous approximation fails.

We remedy this in a two-step process.
First, we define $\gs_\sigma(\l)$, an averaged version of $\gs(\l)$ over an interval of length $\sigma$ around $\l$, for a small $\sigma > 0$.
Because $\gs(\l)$ is bounded by $\abs{\gs(0)}$, this will make the new function $\abs{\gs(0)} / \sigma$-Lipschitz.
Then, to prove that this $\gs_\sigma(\l)$ is useful, we prove that using a slightly incorrect $\gs$ leads to small error in the dual function $f$: $f\qty\big(\l, \gs(\l + \e)) \le R^\star + \order{\abs{\e}}$.
This means that even if we use $\gs_\sigma(\l)$ instead of $\gs(\l)$, we still get a good bound.
However, we emphasize that the previous inequality of using $\gs(\l+\e)$ instead of $\gs(\l)$ is not good enough to get our desired bounds: the perturbation of $\l$ in \cref{eq:dynamic:prop2} is not fixed and depends on the realizations of $\vL, \vO$; the above inequality requires a fixed perturbation.
We now present this property formally.

\begin{lemma} \label{lem:dynamic:nearby_g}
    Let $\gs(\l)$ be as defined in \cref{eq:dynamic:g_def} and fix $\l \ge 0$ and $\l + \e \ge 0$.
    Then,
    \begin{equation*}
        f\qty\big( \l, \gs(\l + \e) )
        \le
        R^\star + \mu \e^+ - \gs(0) \e^-
    \end{equation*}
    where $\e^+ = \max\{\e, 0\}$ and $\e^- = \max\{-\e, 0\}$.
\end{lemma}

\begin{proof}
    Let $\hatvO$ be a bidding function that is $\d$-close to the supremum in the definition of $f\qty\big( \l, \gs(\l + \e) )$, i.e.,
    \begin{equation*}
        f\qty\big( \l, \gs(\l + \e) ) - \d
        \le
        \UO(\hatvO) - \mu \l \PO(\hatvO) + \gs(\l + \e)\qty\big( \rL - \l \PL(\hatvO) )
    \end{equation*}

    By the definition of $\gs(\l + \e)$, we have that
    \begin{equation*}
        R^\star
        \ge
        f\qty\big( \l + \e  , \gs(\l + \e) )
        \ge
        \UO(\hatvO) - \mu (\l + \e) \PO(\hatvO) + \gs(\l + \e)\qty\big( \rL - (\l + \e) \PL(\hatvO) )
    \end{equation*}

    Adding the two inequalities, we get that
    \begin{equation*}
        f\qty\big( \l, \gs(\l + \e) ) - \d - R^\star
        \le
        \mu \e \PO(\hatvO)
        +
        \gs(\l + \e) \e \PL(\hatvO)
        \le
        \mu \e^+
        -
        \gs(0) \e^-
    \end{equation*}
    where in the inequality we used that $\PO, \PL \le 1$ and that $\gs(0) \le \gs(\l + \e) \le 0$.
    Taking $\delta$ arbitrarily close to $0$ proves the lemma.
\end{proof}

We now define the averaged version of $\gs$ that is Lipschitz.
Specifically, for a parameter $\sigma > 0$, we define
\begin{equation} \label{eq:dynamic:g_sigma_def}
    \gs_\sigma(\l)
    =
    \frac{1}{\sigma} \int_\l^{\l+\sigma} \gs(x) dx
\end{equation}
and proceed to prove that $\gs_\sigma$ satisfies all the properties that $\gs$ satisfied, with the addition of being $\order{1/\sigma}$-Lipschitz.

\begin{lemma}[$\gs_\sigma$ Properties] \label{lem:dynamic:g_sigma_props}
    For any $\sigma > 0$, the following properties hold for $\gs_\sigma(\l)$, as defined in \cref{eq:dynamic:g_sigma_def}.
    \SetLabelAlign{CenterWithParen}{\hfil\makebox[.8em]{#1}\hfil}
    \begin{enumerate}[label=(\roman*),align=CenterWithParen]
        \item\label{item:dynamic:g_sigma_increasing}
        It is weakly increasing.
        
        \item\label{item:dynamic:g_sigma_bounds}
        It satisfies: $-\infty < \gs(0) \le \gs_\sigma(\l) \le 0$.
        
        \item\label{item:dynamic:g_sigma_becomes_zero}
        It eventually becomes $0$ if $\gs$ becomes $0$: if $\gs(\l) = 0$ for some $\l$ then $\gs_\sigma(\l) = 0$.
        
        \item\label{item:dynamic:g_sigma_at_f}
        For every $\l \ge 0$, it holds
        \begin{equation*}
            f\qty\big(\l, \gs_\sigma(\l))
            \le
            R^\star
            +
            \mu \sigma
        \end{equation*}
        
        \item\label{item:dynamic:g_sigma_lipschitz}
        It is Lipschitz. Specifically, for every $\l, \e \ge 0$, $0 \le \gs_\sigma(\l + \e) - \gs_\sigma(\l) \le - \frac{\e}{\sigma} \gs(\l)$.
    \end{enumerate}
\end{lemma}

\begin{proof}
    We prove that $\gs_\sigma(\cdot)$ is weakly increasing by examining its derivative, $\frac{\gs(\l + \sigma) - \gs(\l)}{\sigma} \ge 0$, since $\gs(\cdot)$ is weakly increasing.

    For the lower and upper bounds, we have that $\gs_\sigma(\l) \in [\gs(\l), \gs(\l+\sigma)] \sub [\gs(0), 0]$.
    We recall that $\gs(0) > - \infty$ by \cref{lem:dynamic:g_at_zero}.

    Since $\gs(\cdot)$ is weakly increasing and non-positive, if $\gs(\l) = 0$ then $\gs(x) = 0$ for $x \in [\l, \l+\sigma]$, implying $\gs_\sigma(\l) = 0$.

    To prove property \ref{item:dynamic:g_sigma_at_f}, we first recall that the function $f(\l, \cdot)$ is convex (\cref{cl:dynamic:f_convex}).
    Convexity means that when evaluating $f(\l, \gs_\sigma(\l))$, since $\gs(\l) \le \gs_\sigma(\l) \le \gs(\l + \sigma)$, its maximum is at one of the boundaries:
    \begin{alignat*}{3}
        \Line{
            f\qty\big(\l, \gs_\sigma(\l))
        }{\le}{
            \max\qty{\Big.
                f\qty\big(\l, \gs(\l)),
                f\qty\big(\l, \gs(\l + \sigma))
            }
            \le
            R^\star
            +
            \mu\sigma
        }{\textrm{by \cref{lem:dynamic:nearby_g}}}
    \end{alignat*}

    Finally, to prove Lipschitzness, by the mean value theorem, we have that
    \begin{equation*}
        \gs_\sigma(\l + \e) - \gs_\sigma(\l)
        \le
        \e  \sup_{x \in [\l, \l + \e]} \frac{\gs(x + \sigma) - \gs(x)}{\sigma}
        \le
        - \frac{\e}{\sigma} \gs(\l)
        \le
        - \frac{\e}{\sigma} \gs(0)
    \end{equation*}
    which concludes the proof of the lemma.
\end{proof}

\subsubsection{Using \texorpdfstring{$\gs_\sigma(\cdot)$}{g star sigma} in Lemma \ref{lem:dynamic:separation}}

With \cref{lem:dynamic:g_sigma_props}, we have all the properties that we had highlighted, and we are ready to define the $G(\cdot)$ that we will use in \cref{lem:dynamic:separation}.
Specifically, for any $\sigma > 0$, we define
\begin{equation} \label{eq:dynamic:G_def}
    G_\sigma(\l)
    =
    - \int_\l^\infty \gs_\sigma(x) dx
    .
\end{equation}

We note that the above is well-defined: since $\gs_\sigma(\cdot)$ is monotone, it is integrable.
In addition, if $\gs_\sigma(\cdot)$ becomes $0$ eventually, then $G_\sigma(\cdot)$ is always bounded.
We now prove that the above $G_\sigma(\cdot)$ satisfies \cref{lem:dynamic:separation} for an appropriate $A$.

\begin{lemma} \label{lem:dynamic:usage_of_G}
    Let $G_\sigma(\cdot)$ be as defined in \cref{eq:dynamic:G_def} and assume that $\gs(\bar\l) = 0$ for some $\bar\l \ge \rL$.
    Then, $G_\sigma(\cdot)$ satisfies \cref{lem:dynamic:separation} with
    \begin{equation*}
        A =
        R^\star
        +
        \mu\, \sigma
        -
        \eta\frac{\rL^2 + \qty(\frac{\bar\l - \eta \rL}{1 - \eta})^2 }{\sigma}
        \gs(0)
        .
    \end{equation*}
\end{lemma}

Before proving the lemma, we apply it to prove \cref{thm:main}.

\begin{proof}[Proof of \cref{thm:main}]
    By the assumptions of \cref{thm:main} and \cref{lem:dynamic:g_becomes_zero}, we have that $\gs(\bar\l) = 0$ for some finite $\bar\l$ (and by monotonicity in \cref{lem:dynamic:g_monotone}, we may take $\bar\l \ge \rL$ without loss of generality).
    Because of \cref{lem:dynamic:usage_of_G}, the total expected Lagrangian value of the optimizer when the learner starts at $\l^{(1)} \in [0, \rL]$ is
    \begin{alignat*}{3}
        \Line{
            R_T(\l^{(1)})
        }{\le}{
            T R^\star
            +
            \mu\, \sigma \, T
            -
            \eta\frac{\rL^2 + \qty(\frac{\bar\l - \eta \rL}{1 - \eta})^2 }{\sigma}
            \gs(0) T
            +
            \frac{G_\sigma(\l^{(1)})}{\eta}
        }{}
        \\
        \Line{}{\le}{
            T R^\star
            +
            \mu\, \sigma \, T
            +
            \eta\frac{\rL^2 + \qty(\frac{\bar\l - \eta \rL}{1 - \eta})^2 }{\sigma}
            \frac{\Ex{\vO} - R^\star}{\rL} T
            +
            \frac{\bar\l}{\eta} \frac{\Ex{\vO} - R^\star}{\rL}
        }{}
        \\
        \Line{}{=}{
            T R^\star
            +
            \order{
                \sigma \, T
                +
                \frac{\eta}{\sigma} T
                +
                \frac{1}{\eta}
            }
        }{}
    \end{alignat*}
    where the second inequality follows by $G_\sigma(\l^{(1)}) \le G_\sigma(0) \le -\bar\l \gs(0)$ and \cref{lem:dynamic:g_at_zero}; in the last equality we notice that $\bar\l$, $\mu$, $\rL$, $R^\star$, and $\Ex{\vO}$ are constants with respect to $T$, $\eta$, $\sigma$, since their definitions are independent of them.
    Using that the total expected reward the budgeted optimizer gets is at most $R_T(\l^{(1)}) + T \mu \rO$ and that $R^\star + \mu \rO = \opt$, we get that the total expected reward of the budgeted optimizer is at most
    \begin{equation*}
        T \opt
        +
        \order{
            \sigma \, T
            +
            \frac{\eta}{\sigma} T
            +
            \frac{1}{\eta}
        }
    \end{equation*}
    By taking $\sigma = \Theta(\sqrt\eta)$ and then $\eta = \Theta(T^{-2/3})$ we get the final bounds of \cref{thm:main}.
\end{proof}

We now proceed to prove the lemma.

\begin{proof}[Proof of \cref{lem:dynamic:usage_of_G}]
    We notice that $G_\sigma(\l) \ge 0$ since $\gs_\sigma(x) \le 0$ for all $x \ge 0$.
    This takes care of the first condition of \cref{lem:dynamic:separation}.
    We now prove the second condition.
    Let $\e = \eta\qty\big( \rL - \l \PL\qty(\hatvO, \vL, \vO) )$.
    If $\e \ge 0$, we have that
    \begin{equation*}
        G_\sigma(\l + \e) - G_\sigma(\l)
        =
        \int_{\l}^{\l + \e} \gs_\sigma(x) dx
        \le
        \e \, \gs_\sigma(\l + \e)
        \le
        \e \, \qty(
            \gs_\sigma(\l)
            -
            \frac{\e}{\sigma}\gs(\l)
        )
    \end{equation*}
    where we use $\gs_\sigma(x) \le \gs_\sigma(\l + \e)$ in the first inequality and \cref{lem:dynamic:g_sigma_props}-\ref{item:dynamic:g_sigma_lipschitz} in the second inequality.
    If $\e \le 0$ and $\l + \e \ge 0$, we have
    \begin{equation*}
        G_\sigma(\l + \e) - G_\sigma(\l)
        =
        - \int_{\l+\e}^{\l} \gs_\sigma(x) dx
        \le
        \e \, \gs_\sigma(\l + \e)
        \le
        \e \, \qty(
            \gs_\sigma(\l)
            -
            \frac{\e}{\sigma}\gs(\l + \e)
        )
    \end{equation*}
    where we use $\gs_\sigma(x) \ge \gs_\sigma(\l + \e)$ in the first inequality and $\gs_\sigma(\l) - \gs_\sigma(\l + \e) \le \frac{\e}{\sigma}\gs(\l + \e)$ (which again comes from
    \cref{lem:dynamic:g_sigma_props}-\ref{item:dynamic:g_sigma_lipschitz}) in the second inequality.
    We summarize the bound for any $\e$ as
    \begin{equation*}
        G_\sigma(\l + \e) - G_\sigma(\l)
        \le
        \e \, \gs_\sigma(\l)
        -
        \frac{\e^2}{\sigma}\gs\qty\big( \l - \e^- )
    \end{equation*}

    We now use the above inequality in the r.h.s. of \cref{eq:dynamic:prop2}:
    \begin{align*}
        &
        \sup_{\hatvO \in \Rz^{[0, 1]}}\qty(
            \UO\qty(\hatvO)
            -
            \mu \l \PO\qty(\hatvO)
            +
            \frac{1}{\eta}\Ex[\vL, \vO]{\bigg.
                G_\sigma\qty\Big(
                    \l + \eta\qty\big( \rL - \l \PL\qty(\hatvO, \vL, \vO) )
                ) - G_\sigma(\l)
            }
        )
        \\&\quad\le
        \sup_{\hatvO \in \Rz^{[0, 1]}}\Bigg(
            \UO\qty(\hatvO)
            -
            \mu \l \PO\qty(\hatvO)
            +
            \frac{1}{\eta}
            \E_{\vL, \vO} \bigg[
                \eta\qty\big( \rL - \l \PL\qty(\hatvO, \vL, \vO) ) \gs_\sigma(\l)
        \\&\hspace{180pt}
                -
                \frac{\eta^2\qty\big( \rL - \l \PL\qty(\hatvO, \vL, \vO) )^2}{\sigma}
                \gs\qty\Big( 
                    \l
                    -
                    \eta\qty\big( \rL - \l \PL\qty(\hatvO, \vL, \vO) )^-
                )
            \bigg]
        \Bigg)
        \\&\quad\le
        \sup_{\hatvO \in \Rz^{[0, 1]}}\qty(
            \UO\qty(\hatvO)
            -
            \mu \l \PO\qty(\hatvO)
            +
            \gs_\sigma(\l) \qty\big( \rL - \l \PL\qty(\hatvO) ) 
            -
            \frac{\eta\qty\big( \rL^2 + \l^2 )}{\sigma}
            \gs\qty\Big( 
                \l
                -
                \eta\qty\big( \rL - \l  )^-
            )
        )
        \\&\quad=
        f\qty\big( \l, \gs_\sigma(\l) )
        -
        \frac{\eta\qty\big( \rL^2 + \l^2 )}{\sigma}
        \gs\qty\Big( 
            \l
            -
            \eta\qty\big( \rL - \l  )^-
        )
        \le
        R^\star
        +
        \mu\, \sigma
        -
        \frac{\eta\qty\big( \rL^2 + \l^2 )}{\sigma}
        \gs\qty\Big( 
            \l
            -
            \eta\qty\big( \rL - \l )^-
        )
    \end{align*}
    where in the second inequality we used the definition $\Ex{\PL(\hatvO, \vL, \vO)} = \PL(\hatvO)$ and that $\PL(\hatvO, \vL, \vO) \le 1$; in the last inequality we used \cref{lem:dynamic:g_sigma_props}-\ref{item:dynamic:g_sigma_at_f}.
    When $\l \ge \frac{\bar\l - \eta \rL}{1 - \eta}$ and $\bar\l \ge \rL$, we have that $\l - \eta\qty\big( \rL - \l )^- \ge \bar\l$, implying that $\gs(\l - \eta\qty\big( \rL - \l )^-) = 0$ by $\gs(\bar \l) = 0$, $\gs(\cdot) \le 0$, and \cref{lem:dynamic:g_monotone}.
    When $\l \le \frac{\bar\l - \eta \rL}{1 - \eta}$, by taking $\gs(\cdot) \ge \gs(0)$ by \cref{lem:dynamic:g_monotone} and bounding $\l^2$ we get the lemma.
    This completes the proof.
\end{proof}

\printbibliography{}

\clearpage
\appendix
\crefalias{section}{appendix}
\crefalias{subsection}{appendix}
\section{Extended Related Work} \label{sec:app:related}

Here we include a more detailed version of our related work in \cref{sec:related}, focusing on works in budgeted auctions.

\paragraph{Pacing Multipliers in Static Settings}
Pacing multipliers have been studied extensively, since they simplify the action space of bidders in second-price auctions from an arbitrary mapping from values/context to bids.
\cite{BBW15} use pacing multipliers to prove the existence of a Fluid Mean-field Equilibrium in budgeted second price auctions.
\cite{DBLP:journals/ior/ConitzerKSM22} extend this idea and introduce the Second Price Pacing Equilibrium (SPPE), an equilibrium with a finite number of players, all of which are best responding.
\citet{DBLP:journals/mansci/ConitzerKPSMSW22} use pacing multipliers in first-price auctions, where they come up with a weaker equilibrium notion where the players are not necessarily best responding with respect to any bidding function, but are using the optimal pacing multiplier.
\cite{CKK24} prove that computing an SPPE is PPAD-hard.
\cite{BKK23} show that in budgeted first-price auctions, one can get an equilibrium by composing a pacing multiplier with an unbudgeted Bayes-Nash Equilibrium.

\paragraph{Learning in Budgeted Second-price Auctions}
This area was started by the seminal work of \citet{BG19}, who studied budget-constrained learners trying to maximize their total quasi-linear utility.
Their proposed algorithm is to find the optimal pacing multiplier.
From the perspective of a single learner, they prove $\order*{\sqrt T}$ regret in Bayesian environments and $B/T$ competitive ratio in adversarial environments.
They also show that if all the bidders follow their algorithm, under rather strong assumptions on the value distributions, this process converges to an SPPE.
Finally, they show that this behavior is an equilibrium (i.e., players do not want to use a different algorithm), as long as the probability that any two players simultaneously have positive values is small.~\cite{BLM23} show how dual mirror descent based algorithms can simultaneously achieve good guarantees in stochastic and adversarial settings (with adversarial guarantees degrading with how much the adversary is changing the environment).~\cite{BLMS23} uncover a deep connection between dual mirror descent and PID controller algorithms, and provide the first regret guarantees for PID controllers for online resource allocation problems.~\cite{BKMSW23} study budget pacing in settings where the bidder's value distribution changes across rounds, and show how having just one sample per distribution is good enough to achieve $O(\sqrt{T})$ regret.~\cite{FPW23} examine the setting where the learning player also has an ROI constraint, and show how to achieve $\orderT*{\sqrt T}$ regret in Bayesian environments.~\cite{BBFLMSW24} extend this further by showing how to achieve the same with strategies that are fully decoupled across the budget and ROS constraints.

\paragraph{Learning in Budgeted First-price Auctions}
Learning in untruthful auctions poses major challenges compared to truthful ones, since the optimal strategy is no longer a simple pacing multiplier.
Instead, the optimal strategy might involve an arbitrarily complicated mapping from values to bids.
\cite{WYDK23} study budgeted settings when the value and the highest competing bid are independent, and achieve $\orderT*{\sqrt T}$ regret in Bayesian environments.
\cite{CCK24} study first-price auctions with budget and ROI constraints, and show no-regret guarantees in Bayesian and adversarial environments when there is a finite number of possible values and bids.
\cite{AFZ25} extend these results in Bayesian settings for uncountable value and bidding spaces, when the goal is to learn the optimal Lipschitz mapping from values to bids.

\paragraph{Learning in Budgeted First-price Auctions with Pacing Multipliers}
While not optimal, lots of past work has focused on learning pacing multipliers in repeated budgeted auctions and the implications of doing so.
\cite{DBLP:conf/innovations/GaitondeLLLS23} show that if all the players are learning pacing multipliers with a generalized version of the algorithm of~\cite{BG19}, the resulting welfare is at least half of the optimal one in either first- or second-price auctions.
\cite{DBLP:journals/mor/FikiorisT25} generalize this result for first-price auctions and show that, as long as every player has low regret with respect to the optimal pacing multiplier, then the resulting welfare is high.
\cite{DBLP:conf/colt/LucierPSZ24} extend the results of \cite{DBLP:conf/innovations/GaitondeLLLS23} to the case when players also have ROI constraints. 
They also show a simple algorithm similar to ours has $o(T)$ regret against the following benchmark: the value achieved by the optimal sequence of multipliers that observes the average budget constraint in each round in expectation, as long as this sequence does not change too much between rounds.

\paragraph{Learning in Budgeted Environments}
The study of learning in budgeted environments pre-dates the study of learning in budgeted auctions.
\cite{BKS18} introduce \textit{Bandits with Knapsacks}, a generalization of Multi-armed Bandits with $m$ budget constraints, and propose algorithms with $\orderT*{\sqrt T}$ regret in Bayesian environments.
\cite{AD16} offer similar guarantees in the contextual setting.
\cite{ISSS22} extend this study to adversarial environments and show how sublinear regret is not achievable here; instead, they offer $\order*{\log T}$ competitive ratio guarantees that are tight when the budget is sublinear in $T$.
\cite{KS20} improve the competitive ratio of~\cite{ISSS22} by improving the dependence on the number of budget constraints.
\cite{CCK22} offer $\order*{B/T}$-competitive ratio that is optimal when $B = \Omega(T)$ and also extend these results when the budgets are replenishable (which models an ROI constraint that can both increase or decrease) in \cite{CCK24}.
Finally, \cite{BLMSX25} develop a weaker benchmark against which they achieve sublinear regret.

\section{Budgeted Stackelberg Equilibrium for multiple optimizer constraints} \label{sec:app:defs}

In this section, we offer the generalized definition of \cref{def:defs:budgeted_stackelberg}, where the Leader has multiple constraints.

\begin{definition}[Budgeted Stackelberg Equilibrium] \label{def:defs:budgeted_stackelberg_mult_constraints}
    Consider strategy spaces of two players, $\calA$ for the optimizer and $\mathcal B$ for the learner.
    For actions $a \in \calA$ and $b \in \mathcal B$ let
    \begin{itemize}
        \item $U(a, b) \in \R$ be the utility of the optimizer,
        \item $P_i(a, b) \in \R$ be the payment of the optimizer for a resource $i \in [m]$ with average budget $\rho_i \in \R$.
    \end{itemize}

    Also, for a distribution of the Optimizer's actions $A \in \Delta(\calA)$, let $BR(A) \sub \Delta(\mathcal B)$ be the distributions of actions of the learner that are best responses to $A$.
    Then, a \textit{Budgeted Stackelberg Equilibrium} is a distribution $\nu$ over $\Delta(\calA) \times \Delta(\mathcal B)$, such that the expected utility of the Optimizer is maximized while satisfying all the constraints in expectation, while the learner is best responding for every random action of the Optimizer.
    Formally, the value of the Budgeted Stackelberg Equilibrium is
    \begin{optimization} \label{eq:defs:budgeted_stackelberg_mult_constraints}
        \sup_{\nu \in \Delta\qty\big( \Delta(\calA) \times \Delta(\mathcal B) )}
        \quad &
        \Ex[(A, B) \sim \nu]{
            \Ex[a \sim A, b\sim B]{U(a, b)}
        }
        \\
        \textrm{\textup{such that }}
        \quad &
        \Ex[(A, B) \sim \nu]{
            \Ex[a \sim A, b\sim B]{P_i(a, b)}
        } \le \rho_i \qquad \forall i \in [m]
        \\[10pt] & \quad
        B \in BR(A) \qquad\qquad \forall (A, B) \in \textrm{supp}(\nu)
    \end{optimization}
\end{definition}

And we now offer the generalization of \cref{lem:defs:budgeted_stackelberg_simple}, that we can limit the support of distribution $\nu$ to at most $m + 1$ points.

\begin{lemma} \label{lem:defs:budgeted_stackelberg_simple_mult_constraints}
    The value of \cref{eq:defs:budgeted_stackelberg_mult_constraints} is equal to
    \begin{optimization} \label{eq:defs:budgeted_stackelberg_simple_mult_constraints}
        \sup_{\substack{
            A_j \in \Delta(\calA), \; z_j \ge 0,
            \\
            B_j \in \Delta(\mathcal B)
        }}
        \quad &
        \sum_{j = 1}^{m+1} z_j \Ex[a \sim A_j, b \sim B_j]{U(a, b)}
        \\
        \textrm{\textup{such that }}
        \quad &
        \sum_{j = 1}^{m+1} z_j \Ex[a \sim A_j, b \sim B_j]{P_i(a, b)} \le \rho_i \qquad \forall i \in [m]
        \\[10pt] &
        \sum_{j = 1}^{m+1} z_j = 1
        \\[10pt] &
        B_j \in BR(A_j) \qquad\qquad \forall j \in [m+1]
    \end{optimization}
\end{lemma}

\begin{proof}
    Define $U(A, B) = \Ex[a \sim A, b \sim B]{U(a, b)}$ and $P_i(A, B) = \Ex[a \sim A, b \sim B]{P_i(a, b)}$ for $i \in [m]$.
    Consider the following subset of $\R^{m+1}$:
    \begin{equation*}
        F = \qty{
            \qty\big(U(A, B), P_1(A, B), P_2(A, B), \ldots, P_m(A, B))
            :
            B \in \BR(A), A \in \Delta(\calA)
        }
    \end{equation*}
    its closure $\bar F$, and the convex hull of the closure, $\conv(\bar F)$.
    Then, \cref{eq:defs:budgeted_stackelberg_mult_constraints} can be re-written as
    \begin{optimization*}
        \sup_{x \in \conv(\bar F)} \;
        \quad &
        x[1]
        \\
        \textrm{ such that }
        \quad &
        x[i + 1] \le \rho_i \qquad \forall i \in [m]
    \end{optimization*}

    Via Carath\'eodory's theorem, any point in $\conv(\bar F)$ can be written as a convex combination of at most $m+2$ points of $\bar F$, since $F$ has $m+1$ dimensions.
    We now proceed to show that we actually need $m+1$ points for the optimal distribution.
    Consider a collection of $m+2$ points $x_1, x_2, \ldots, x_{m+2}$ and consider we are trying to find the optimal way to mix them.
    This leads to the following linear program over $m+2$ variables:
    \begin{optimization*}
        \max_{z_1, \ldots, z_{m+2} \ge 0}
        \quad &
        \sum_{j=1}^{m+2} z_j x_j[1]
        \\
        \textrm{ such that }
        \quad &
        \sum_{j=1}^{m+2} z_j x_j[i+1] \le \rho_i \qquad \forall i \in [m]
        \\ &
        \sum_{j=1}^{m+2} z_j = 1
    \end{optimization*}

    Since this is a linear program with $m+1$ constraints, if feasible, its optimal solution can always be supported in $m+1$ variables, i.e., in the optimal solutions one of the $z_j$'s would be $0$.
    This means that we can always drop one variable from any potential solution, i.e., the optimal solution can be written as a convex combination of at most $m+1$ points of $\bar F$.
    This completes the proof.
\end{proof}

\section{Omitted Proofs of Section \ref{sec:auctions}} \label{sec:app:auctions}

In this section, we include the omitted proofs of \cref{sec:auctions}, of \cref{lem:auction:learner_budget,lem:auction:strong_duality}.

\AlgorithmProps*

\begin{proof}
    To prove that the bids are non-negative, we notice that
    \begin{equation}\label{eq:531}
        \l^{(t+1)}
        =
        \l^{(t)}
        +
        \eta\qty(\rL - p_{\mathtt L}^{(t)})
        \ge
        \l^{(t)}
        +
        \eta\qty(\rL - \l^{(t)})
        =
        (1-\eta)\l^{(t)}
        +
        \eta \rL
        \ge
        \min\qty{\l^{(t)}, \rL}
    \end{equation}
    where the first inequality follows from the payment rule in either first- or second-price and that $\vL \le 1$; the second inequality follows from $0 < \eta < 1$.
    Therefore, if $\l^{(1)} \ge 0$, then $\l^{(t)} \ge 0$ in every round $t$.

    By summing \eqref{eq:update} over all $t \in [T]$ we get that $\sum_{t \in [T]} p_{\mathtt L}^{(t)} = \rL T + \frac{\l^{(1)} - \l^{(T+1)}}{\eta}$ which proves the payment part of the claim.
    We get that the learner does not run out of budget by \cref{eq:531}, which recursively proves $\l^{(T+1)} \ge \min\{ \l^{(1)}, \rL \} = \l^{(1)}$, if $\l^{(1)} \le \rL$.
\end{proof}

\StrongDuality*

\begin{proof}
    We first simplify \cref{opt:auctions:stackelberg_eq}.
    Let $F \sub \R_{\ge 0}^2$ be all the pairs of expected utility and payment the optimizer can get while the learner is best responding:
    \begin{equation*}
        F = \qty{
            \qty( \UO(\hatVO), \l\, \PO(\hatVO) )
            :
            \hatVO \in \Delta\qty(\Rz^{[0,1]}), \l \in \Rz,
            \textrm{ and } \l\, \PL(\hatVO) \ge \rL
        }
    \end{equation*}

    Let $\bar F$ be the closure of $F$ and let $\conv(\bar F)$ be its convex hull.
    Then \cref{opt:auctions:stackelberg_eq} corresponds to
    \begin{optimization} \label{eq:auction:opt_prob_simple}
        \opt =
        \max_{x \in \conv(\bar F)}
        \quad & x[1]
        \\
        \textrm{such that }
        \quad &
        x[2] \le \rO
    \end{optimization}

    Then, if strong duality holds and the above problem always has a feasible solution, we have that the dual and the primal have the same value, and there exists a finite Lagrange multiplier that achieves this, i.e.,
    \begin{optimization*}
        \opt =
        \min_{\mu \ge 0} \max_{x \in \bar F}
        \quad x[1] - \mu \, x[2] + \mu \, \rO
        .
    \end{optimization*}

    We get that we can maximize $x$ over $F$ instead of $\conv(\bar F)$ by noticing that the problem becomes linear and unconstrained, so the optimal value is achieved at a boundary point of $\conv(\bar F)$ which is included in $\bar F$.

    For strong duality to hold, we only need to show that there exists a point $x$ in the relative interior of $\conv(\bar F)$ such that $x[2] \le \rO$ \cite[Proposition 5.3.1]{bertsekas2009convex}.
    Consider the optimizer bidding using
    \begin{equation*}
        \hatVO =
        \begin{cases}
            0, &\quad\textrm{ w.p. } q \\
            \e, &\quad\textrm{ w.p. } (1-q)(1-\e^2) \\
            2, &\quad\textrm{ w.p. } (1-q)\e^2
        \end{cases}
    \end{equation*}
    for some $\e\in (0, 1)$ and $q \in (0, 1/2)$.
    Since the learner's value has a positive expectation, we can always take $\e$ small enough such that $\Pr{\vL \ge \e} > 0$ and therefore $\PL(\e) > 0$, where $\PL(\e)$ is the learner's expected payment when the optimizer bids $\e$ (and similarly for $\PO(\e), \UO(\e)$).
    Note also that bid $0$ always loses (as ties go to the learner), so $\UO(0) = \PO(0) = 0$ and $\PL(0) \ge 0$.
    Then, using $\l_0(q) = \frac{\rL}{ (1 - q)(1 - \e^2) \PL(\e) }$ and any $\l \ge \l_0(q)$ we get
    \begin{align*}
        \UO(\hatVO) &= (1-q)\qty\Big( (1-\e^2) \Ex{\vO \One{\vL < \e}} + \e^2 \Ex{\vO} )
        , \\
        \l_0(q)\PO(\hatVO) &= \rL \frac{ (1-\e^2) \PO(\e) + \e^2 \PO(2) }{ (1-\e^2) \PL(\e) }
        , \\
        \l \PL(\hatVO) &\ge \l_0(q) (1-q)(1-\e^2) \PL(\e) = \rL
        .
    \end{align*}
    Notice that $y_\e \coloneqq \l_0(q)\PO(\hatVO)$ does not depend on $q$, is strictly positive (since $\PO(2) \ge \Ex{\vL} > 0$), and satisfies $y_\e \to 0$ as $\e \to 0$: in second-price auctions, $\PO(\e) \le \e \Pr{0 < \vL < \e}$, $\PO(2) \le 1$, and $\PL(\e) = \e \Pr{\vL \ge \e}$ give $y_\e \le \rL \frac{(1-\e^2)\Pr{0 < \vL < \e} + \e}{(1-\e^2)\Pr{\vL \ge \e}} \to 0$; in first-price auctions, $\PO(\e) \le \e$, $\PO(2) = 2$, and $\PL(\e) = \Ex{\vL \One{\vL \ge \e}} \to \Ex{\vL} > 0$ give $y_\e \le \rL \frac{\e + 2\e^2}{(1-\e^2)\Ex{\vL \One{\vL \ge \e}}} \to 0$.
    Fix $\e \in (0, 1)$ small enough so that $0 < y_\e \le \rO / 2$, and let $u_\e \coloneqq (1-\e^2)\Ex{\vO \One{\vL < \e}} + \e^2 \Ex{\vO} \ge \e^2 \Ex{\vO} > 0$.
    Then varying $q \in (0, 1/2)$ and setting $\l = c \, \l_0(q)$ for $c \in [1, 2]$ retains feasibility ($\l \PL(\hatVO) \ge \rL$ and $\l \PO(\hatVO) = c \, y_\e \le \rO$) while sweeping out the rectangle $(\frac{1}{2} u_\e, u_\e) \times [y_\e, 2 y_\e] \sub F \sub \conv(\bar F)$ of positive measure in $\R^2$.
    Hence there exists a feasible point in the relative interior of $\conv(\bar F)$, completing the proof for both first- and second-price auctions.
\end{proof}

\section{Second-price Budgeted Stackelberg Equilibrium with multiple Optimizer strategies}
\label{sec:app:multiple_dists_example}

In this section, we present an example of a Budgeted Stackelberg Equilibrium where the optimizer gets strictly more utility by mixing between multiple strategies.

We consider the following value distributions: the optimizer's value is $\vO = 1$ with probability $1$, while the learner's value is $\vL = 1/2$ with probability $1/3$ and $\vL = 1$ with probability $2/3$.
Since the optimizer has a deterministic value, her fake value $\hatvO$ is just a non-negative number instead of a function.
This leads to the following functions for the optimizer's expected utility and payment and the learner's payment if the optimizer uses $\hatvO \ge 0$:
\begin{equation*}
    \UO(\hatvO)
    =
    \begin{cases}
        0, &\textrm{ if } \phantom{1 \le\;} \hatvO \le \frac{1}{2} \\
        \frac{1}{3}, &\textrm{ if } \frac{1}{2} < \hatvO \le 1 \\
        1, &\textrm{ if } 1 < \hatvO
    \end{cases}
    ,\quad
    \PO(\hatvO)
    =
    \begin{cases}
        0, &\textrm{ if } \phantom{1 \le\;} \hatvO \le \frac{1}{2} \\
        \frac{1}{6}, &\textrm{ if } \frac{1}{2} < \hatvO \le 1 \\
        \frac{5}{6}, &\textrm{ if } 1 < \hatvO
    \end{cases}
    ,\quad
    \PL(\hatvO)
    =
    \begin{cases}
        \hatvO, &\textrm{ if } \phantom{1 \le\;} \hatvO \le \frac{1}{2} \\
        \hatvO\frac{2}{3}, &\textrm{ if } \frac{1}{2} < \hatvO \le 1 \\
        0, &\textrm{ if } 1 < \hatvO
    \end{cases}
\end{equation*}
which we also show visually in \cref{fig:app:example_details} (left).
We can assume that any optimal strategy by the optimizer uses only $\hatvO \in \{\frac{1}{2}, 1, 2\}$, since $1$ achieves the same utility and payment compared to any other $\frac{1}{2} < \hatvO \le 1$ but maximizes the Learner's payment.
A similar argument proves that $\frac{1}{2}$ is optimal compared to any $\hatvO \le \frac{1}{2}$, and $2$ is picked arbitrarily from all $\hatvO > 1$.

We now analyze how much utility the Optimizer can achieve if they use a single fixed bidding strategy when their budget is some $\rho \ge 0$ and the learner has a budget of $1$.
If this utility function is strictly convex in any region, then there is a budget for the optimizer where in the Budgeted Stackelberg Equilibrium they are incentivized to randomize between multiple action distributions.
Since the optimizer is picking a distribution over $\hatvO \in \{\frac{1}{2}, 1, 2\}$, the optimization problem we have to solve is
\begin{optimization*}
    \widehat\opt(\rho)
    =
    \max_{\substack{\l, p_1, p_2, p_3 \ge 0\\p_1+p_2+p_3 = 1}}
    \quad &
    p_1\, \UO(\frac{1}{2}) + p_2\, \UO(1) + p_3\, \UO(2)
    \\
    \textrm{ such that }
    \quad &
    \l\,\qty( p_1\, \PO(\frac{1}{2}) + p_2\, \PO(1) + p_3\, \PO(2) ) \le \rho
    \\ &
    \l\,\qty( p_1\, \PL(\frac{1}{2}) + p_2\, \PL(1) + p_3\, \PL(2) ) = 1
\end{optimization*}
which has optimal solution
\begin{equation*}
    \widehat\opt(\rho) =
    \begin{cases}
        \frac{\rho}{1 - \rho}, &\textrm{ if } \rho \le \frac{1}{4}
        \\
        1 - \frac{5}{6(1 + \rho)}, &\textrm{ if } \rho \ge \frac{1}{4}
    \end{cases}
    ,\qquad
    \m{\l\\p_1\\p_2\\p_3} = 
    \begin{cases}
        \qty[ 2(1 - \rho) \quad 1 - \frac{3\rho}{1 - \rho} \quad \frac{3\rho}{1 - \rho} \quad 0 ]^\top, &\textrm{ if } \rho \le \frac{1}{4}
        \\[12pt]
        \qty[ \frac{6(1+\rho)}{5} \quad 0 \quad \frac{5}{4(1 + \rho)} \quad 1 - \frac{5}{4(1 + \rho)} ]^\top, &\textrm{ if } \rho \ge \frac{1}{4}
    \end{cases}
    .
\end{equation*}

We also plot the function $\widehat\opt(\rho)$ in \cref{fig:app:example_details} (right).
We notice that for $\rho < \frac{1}{4}$, the function $\widehat\opt(\rho)$ is strictly convex, since its second derivative is $\frac{2}{(1 - \rho)^3} > 0$.
This means that in that range the optimizer should mix between two bidding strategies to achieve their optimal Budgeted Stackelberg Equilibrium value.
Therefore, for $0 < \rho < \frac{1}{4}$, the optimizer, instead of using the optimal strategy of $\widehat\opt(\rho)$, should do with probability $1 - 4\rho$ the optimal strategy of $\widehat\opt(0)$ and with probability $4\rho$ the optimal strategy of $\widehat\opt(\frac{1}{4})$.

\begin{figure}[t!]
    \centering
    \includegraphics[width=\linewidth]{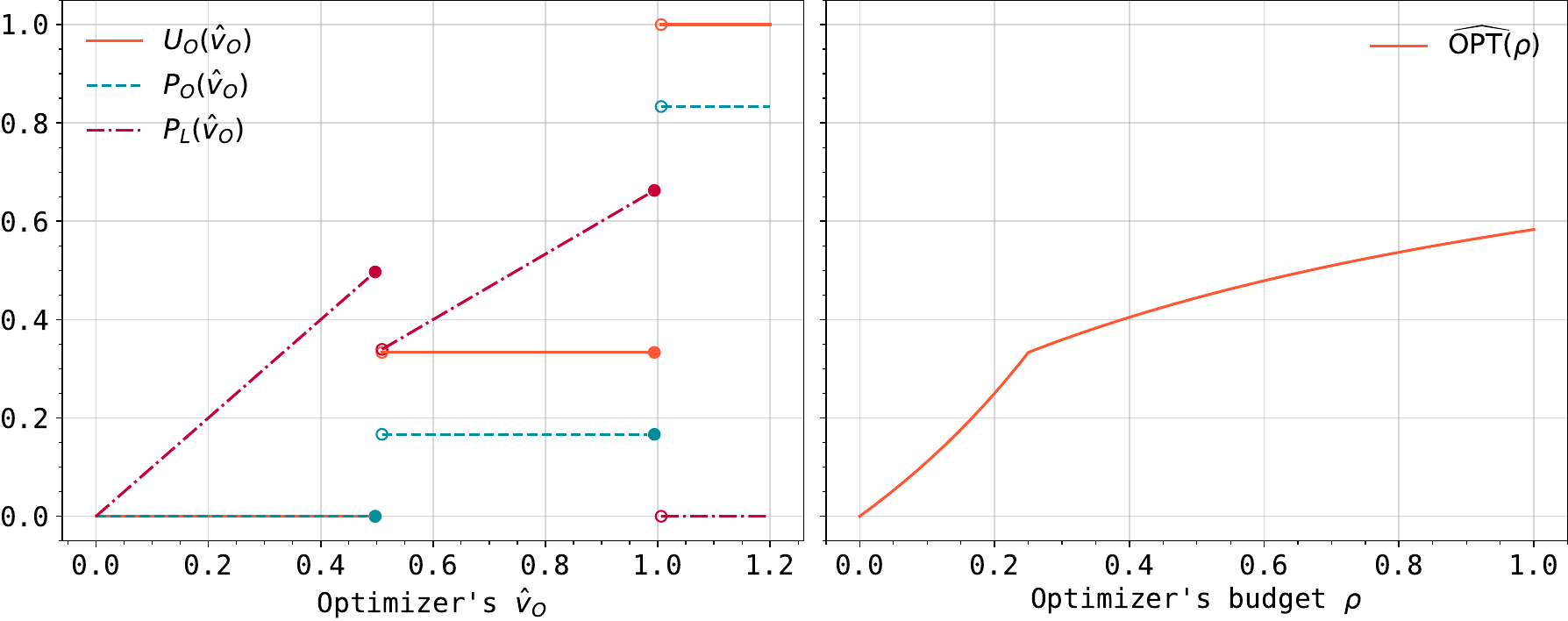}
    \caption{The functions $\UO(\cdot)$, $\PO(\cdot)$, $\PL(\cdot)$, of the example of \cref{sec:app:multiple_dists_example} (left) and the optimizer's optimal value when she uses a single distribution of actions with a budget of $\rho$ (right).}
    \label{fig:app:example_details}
\end{figure}

\section{Example where the Optimizer can get almost linearly more utility than the Equilibrium without the assumption of Theorem \ref{thm:main}}
\label{sec:app:dynamic:example}

In this section, we show an example where the optimizer can get almost $\Theta(T)$ more utility than what the BSE suggests when the learner has large probability mass around value $0$.

The optimizer's value is $\vO = 1$ with probability $1$.
The learner's value is sampled from the following CDF
\begin{equation*}
    F_L(x) =
    \begin{cases}
        \frac{1}{2} (2 x)^\d, & \textrm{ if } x \le \frac{1}{2} \\
        \frac{1}{2}, & \textrm{ if } \frac{1}{2} \le x < 1 \\
        1, & \textrm{ if } 1 \le x
    \end{cases}
\end{equation*}
for some small $\d > 0$.
In other words, with probability $1/2$, $\vL \in [0, 1/2]$ and with probability $1/2$, $\vL = 1$.
When the optimizer bids $\hat v$, we have the following functions
\begin{equation*}
    \UO(\hat v)
    =
    \begin{cases}
        \frac{1}{2} (2 \hat v)^\d, & \textrm{ if } \hat v \le \frac{1}{2} \\
        \frac{1}{2}, & \textrm{ if } \frac{1}{2} \le \hat v \le 1 \\
        1, & \textrm{ if } 1 < \hat v
    \end{cases}
    ,\quad
    \PO(\hat v)
    =
    \begin{cases}
        \frac{\d \, \hat v^{\d+1}}{2^{1-\d}(\d+1)}, & \textrm{ if } \hat v \le \frac{1}{2} \\
        \frac{\d}{4(1 + \d)}, & \textrm{ if } \frac{1}{2} \le \hat v \le 1 \\
        \frac{\d}{4(1 + \d)} + \frac{1}{2}, & \textrm{ if } 1 < \hat v
    \end{cases}
    ,\quad
    \PL(\hat v)
    =
    \hat v \qty\big( 1 - \UO(\hat v) )
\end{equation*}

Let the budget shares be $\rL = \PL(1) = \frac{1}{2}$ and $\rO = \frac{1}{2} \PO(1) = \frac{\d}{8(1 + \d)}$.
We will prove that the following is the optimal solution to the Budgeted Stackelberg equilibrium:
\begin{equation*}
    \begin{cases}
        \qty(\l_1^\star, \hat v_1^\star) = (1, 1), & \textrm{ w.p. } \frac{1}{2}
        \\
        \qty(\l_2^\star, \hat v_2^\star) = (+\infty, 0), & \textrm{ w.p. } \frac{1}{2}
    \end{cases}
\end{equation*}
where by $(\l_2^\star, \hat v_2^\star) = (+\infty, 0)$ we mean the limit of $(\l, \hat v) \to (\infty, 0)$.
We first notice that this is indeed a feasible strategy.
For both pairs the learner is best responding: $\l_1^\star \PL(\hat v_1^\star) = \PL(\hat v_1^\star) = \rL$ and for $(\l_2^\star, \hat v_2^\star)$ for every positive $\e > 0$, it holds $\PL(\e) > 0$, implying that there exists a limiting sequence that satisfies the learner's budget constraint.
As for the optimizer's budget, their spending is $\frac{1}{2} \l_1^\star \PO(\hat v_1^\star) = \frac{1}{2} \PO(1) = \rO$.
Then, the optimizer's utility is
\begin{equation*}
    \opt = \frac{1}{2} \UO(1) = \frac{1}{4}
\end{equation*}

We now show that this is indeed optimal.
To do so, we consider the Lagrangian relaxation of the optimizer's constraint using multiplier $\mu = \frac{\opt}{\rO} = 2 \frac{1 + \d}{\d}$ and its value $R^\star$, as defined in \cref{lem:auction:strong_duality}.
Since $\mu \rO = \opt$, all we have to prove is that $R^\star \le 0$ (equality is achieved by examining our solution above), where
\begin{align*}
    R^\star &=
    \sup_{
        \hatVO \in \Delta(\Rz), \,
        \l \in \Rz
    }
    \quad
    \UO\qty\big(\hatVO) - \mu \l \PO\qty\big(\hatVO)
    \qquad
    \textrm{\textup{ such that }}
    \qquad
    \lambda \PL\qty\big(\hatVO) \ge \rL
    \\&=
    \sup_{
        \hatVO \in \Delta(\Rz), \,
        \l \in \Rz
    }
    \UO\qty\big(\hatVO) - \mu \rL \frac{\PO\qty\big(\hatVO)}{\PL\qty\big(\hatVO)}
    .
\end{align*}

Since all we have to show is that $R^\star \le 0$, all we need to show is that for all $\hatVO$ and $\l$ it holds that
\begin{equation*}
    \UO\qty\big(\hatVO) - \mu \rL \frac{\PO\qty\big(\hatVO)}{\PL\qty\big(\hatVO)}
    \le
    0
    \quad\iff\quad
    \UO\qty\big(\hatVO) \PL\qty\big(\hatVO) \le \mu \rL \PO\qty\big(\hatVO)
\end{equation*}

We prove the above inequality using the following case analysis:
\begin{itemize}
    \item If we restrict to $\hatVO \in \Delta([0, 1/2])$, then $\UO(\cdot)$, $\PL(\cdot)$ are concave and $\PO(\cdot)$ is convex.
    This means that what we have to prove is that for every $\hatvO \in [0, 1/2]$ it holds
    \begin{equation*}
        \UO\qty\big(\hatvO) \PL\qty\big(\hatvO)
        \le
        \mu \rL \PO\qty\big(\hatvO)
        \iff
        2^{-2 + \d} \hatvO^{1 + \d} (2 - 2^\d \hatvO^\d) \le 2^{-1 + \d} \hatvO^{1 + \d}
    \end{equation*}
    which we can easily check is true.

    \item Consider that $\hatVO \in [0, 1/2]$ with probability $q_1$, achieving expected value $\hatvO$ conditioned on that interval.
    Then the only other optimal values are $1$ or anything above $1$, e.g., $2$.
    Assume that $\hatVO = 1$ with probability $q_2$ and $\hatVO = 2$ with probability $q_3$.
    Then, using our observation about convexity/concavity on $[0, 1/2]$, all we have to prove is that
    \begin{equation*}
        \qty\big(q_1 \UO(\hatvO) + q_2 \UO(1) + q_3 \UO(2))
        \qty\big(q_1 \PL(\hatvO) + q_2 \PL(1) + q_3 \PL(2))
        \le
        \mu \rL 
        \qty\big(q_1 \PO(\hatvO) + q_2 \PO(1) + q_3 \PO(2))
    \end{equation*}
    which, using $q_3 = 1 - q_1 - q_2$ and substituting $\UO, \PL, \PO$, we can re-write as
    \begin{align*}
        \d \;& \Big( \;
            q_1 q_2 \qty(-1 + 2 \hatvO + (2\hatvO)^\d - 3 \hatvO (2\hatvO)^\d)
            \\ & \;\; +
            (1 - q_1 - q_2) (q_2 - 3 + 4 \hatvO q_1 (1 - (2\hatvO)^\d))
            \\ & \;\; -
            q_1^2 \hatvO (2\hatvO)^{2\d}
        \\ & \; \Big) 
        - 2 (1 - q_1 - q_2)
        \le 0
    \end{align*}

    We prove the above by showing that the term in each line is non-positive.
    First, we notice that $-1 + 2 \hatvO + (2\hatvO)^\d - 3 \hatvO (2\hatvO)^\d$ is concave for $\hatvO \ge 0$, implying the maximum is attained where the derivative is $0$.
    By evaluating the derivative we see it is $0$ when $\hatvO = 1/2$, which makes the maximum $0$ and implies $-1 + 2 \hatvO + (2\hatvO)^\d - 3 \hatvO (2\hatvO)^\d \le 0$.
    
    Second, the term $4 \hatvO (1 - (2\hatvO)^\d)$ is also concave with respect to $\hatvO$ and its maximum is at $\hatvO = \frac{1}{2}(1 + \d)^{-1/\d}$, making the term $4 \hatvO (1 - (2\hatvO)^\d) \le 2 \d (1 + \d)^{- \frac{1+\d}{\d}} \le 1/2$.
    This makes the term of the second line at most $(1 - q_1 - q_2) (q_2 - 3 + q_1/2)$, which is non-positive since $q_1 + q_2 \le 1$.
    
    The third and fourth terms are both non-positive, making the entire expression non-positive.
\end{itemize}

\begin{figure}[t!]
    \centering
    \includegraphics[width=.99\linewidth]{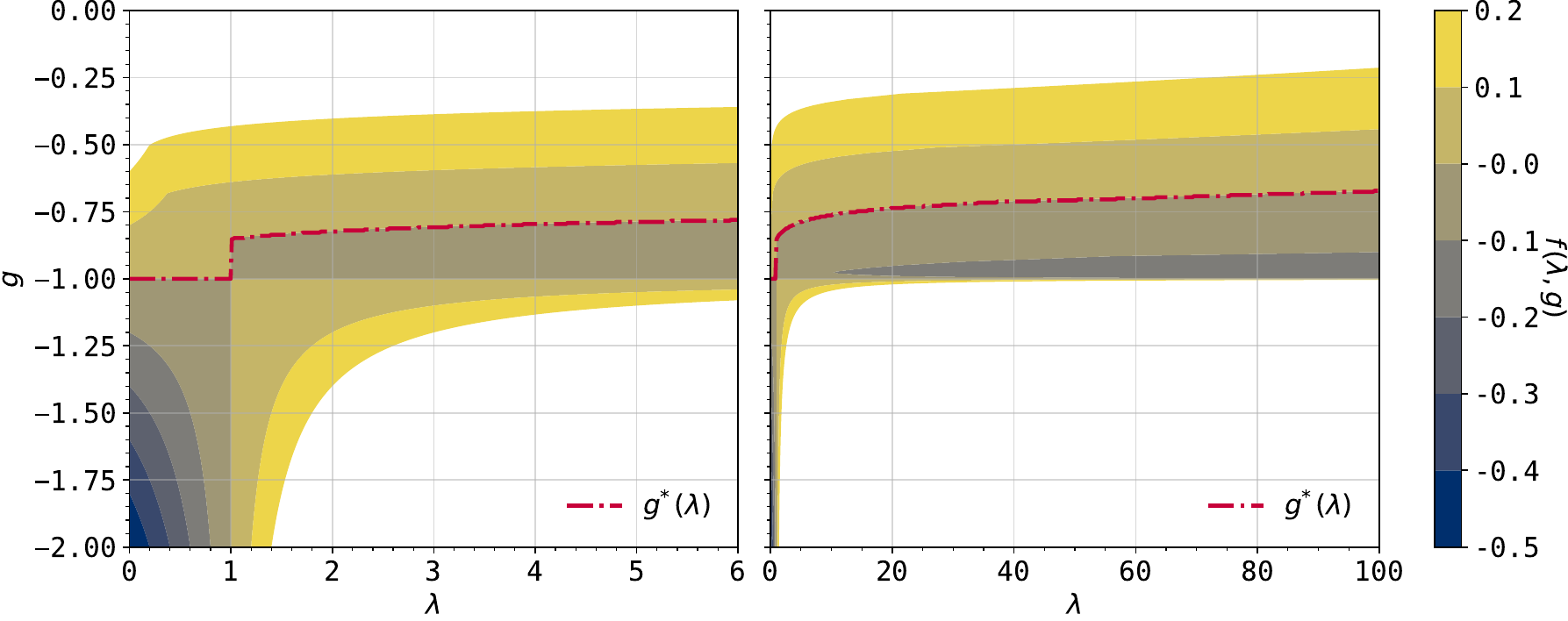}%
    \caption{The dual $f(\l, g)$ function of the example of \cref{sec:app:dynamic:example}, along with the induced $\gs(\l)$ function, for $\d = 0.05$. Plotted for small and large $\l$. We note that $\gs(\l)$ goes to $0$ at a very slow rate, indicating the optimizer's ability to gain substantial utility as $\l^{(t)}$ increases in the repeated game.}
    \label{fig:app:dynamic:f}
\end{figure}

Now consider the repeated setting over $T$ rounds.
As we will show, the optimizer here can achieve more than $\opt \cdot T + \order*{T^{2/3}}$ reward, unlike our guarantee in \cref{thm:main}.
This is because our distributional assumption, that there exists $\e > 0$ such that $\Pr{0 < \vL < \e\vO} = \Pr{0 < \vL < \e}$ fails.
In fact, for $\e \le 1/2$ we have that $\Pr{0 < \vL < \e} = \frac{1}{2} (2\e)^\d$, i.e., this goes to $0$ at a very slow rate.
The more technical reason why this example fails is that \cref{lem:dynamic:g_becomes_zero} fails: $\gs(\l)$ remains strictly negative for all $\l$.
We note that in this example, the second condition of that lemma also fails: it holds that $R^\star = \Ex{\vO \One{\vL = 0}} = 0$.
In \cref{fig:app:dynamic:f} we see that $\gs(\l)$ converges to $0$ at a very slow rate.

To achieve $\opt \cdot T$ reward, the optimizer should bid $\l^{(t)} \hat v_1^\star$ in rounds $t \le T/2$, and then start bidding $0$.
For rounds $t \le T/2$, the learner's payment is $\l^{(t)} \PL(\hat v_1^\star) = \frac{1}{2} \l^{(t)}$.
This makes the expected change in the learner's $\lambda$ equal to $\eta( \rL - \l^{(t)} \PL(\hat v_1^\star) ) = \frac{\eta}{2}(1 - \l^{(t)})$, implying the Learner will stabilize around $\l^{(t)} \approx 1 = \l^\star_1$ in the first $T/2$ rounds, as suggested by the Budgeted Stackelberg equilibrium.
This implies that the optimizer is going to approximately satisfy her budget constraint, leading to the desired $\opt \cdot T$ reward.

However, consider the optimizer's strategy to switch from bidding $\l^{(t)} \hat v_1^\star$ to bidding $\nicefrac{1}{\mu}$ (i.e., use $\hat v = \nicefrac{1}{\l^{(t)}\mu}$) (recall the Lagrange multiplier $\mu = 2 \frac{1 + \d}{\d}$) after the first $T/2 - \tau$ rounds for some $\tau$.
We note that the second strategy is not arbitrary: it is the one that maximizes the optimizer's Lagrangian reward when ignoring the learner's budget constraint.
We now analyze this strategy when everything happens in expectation, i.e., the Learner's payment is exactly $\PL(\nicefrac{1}{\l^{(t)}\mu})$ in round $t$; this is accurate, since as $T$ grows larger, $\eta$ grows smaller and deviations due to random events become irrelevant.
We then complement this analysis with experimental data in \cref{ssec:app:dynamic:experiment}.

Let us calculate the learner's spending: when at $\l^{(t)}$, the learner will increase her pacing multiplier in expectation by
\begin{align} \label{eq:app:6}
    \l^{(t+1)} - \l^{(t)}
    & =
    \eta\qty( \rL - \l^{(t)} \PL\qty(\frac{1}{\l^{(t)}\mu}) )
    =
    \eta \qty(
        \frac{\d^{1 + \d}}{4 (\d+1)^{\d + 1} (\l^{(t)})^\d}
        +
        \frac{1}{2 (\d+1)}
    )
\end{align}
Since $\l^{(T/2 - \tau)} = 1$, and the above difference is always non-negative, we have that for round $t \ge T/2 - \tau$, it holds that
\begin{equation*}
    1 + \qty(t - \frac{T}{2} + \tau) \frac{\eta}{2 (\d+1)}
    \le
    \l^{(t)}
    \le
    1 + \qty(t - \frac{T}{2} + \tau) \frac{\eta}{2}\qty(1 - \frac{\d}{2} - \order{\d^2 \log\frac{1}{\d}})
\end{equation*}

We now analyze the optimizer's expected payment
\begin{alignat*}{3}
    \Line{
        \sum_{t = T/2 - \tau + 1}^T \l^{(t)} \PO\qty(\frac{1}{\l^{(t)}\mu})
    }{=}{
        \frac{\d^{2 + \d}}{ 4(1 + \d)^{2+\d}}
        \sum_{t = T/2 - \tau + 1}^T (\l^{(t)})^{-\d}
    }{}
    \\
    \Line{}{\le}{
        \frac{\d^{2 + \d}}{ 4(1 + \d)^{2+\d}}
        \sum_{t = 1}^{T/2 + \tau} \qty(1 + t \frac{\eta}{2 (\d+1)})^{-\d}
    }{}
    \\
    \Line{}{\le}{
        \frac{\d^{2 + \d}}{ 4(1 + \d)^{2+\d}}
        \int_{t = 0}^{T/2 + \tau} \qty(1 + t \frac{\eta}{2 (\d+1)})^{-\d}
    }{}
    \\
    \Line{}{=}{
        \frac{\d^{2 + \d}}{ 4(1 + \d)^{2+\d} }
        \frac{\qty(1 + \qty(\frac{T}{2} + \tau) \frac{\eta}{2 (\d+1)})^{1 - \d} - 1}{ (1 - \d) \frac{\eta}{2 (1 + \d)} }
    }{}
    \\
    \Line{}{\approx}{
        \qty( \frac{\d^2}{8}  - \order{\d^3 \log\frac{1}{\d}} )
        \frac{T^{1-\d}}{\eta^\d}
    }{}
\end{alignat*}
where in the last approximation we assumed that $\tau = o(T)$, $\eta T = \omega(1)$ and take $\delta \approx 0$.
Since the spending of the optimizer in the first $T/2 - \tau$ rounds is $\l_1^\star \PO(1) = \frac{\d}{4(1+\d)}$ per round and her total budget is $\rO T = \frac{\d}{8(1+\d)} T$ we pick $\tau$ so that
\begin{equation*}
    \frac{\d}{4(1+\d)}
    \qty( \frac{T}{2} - \tau)
    +
    \qty( \frac{\d^2}{8}  - \order{\d^3 \log\frac{1}{\d}} )
    \frac{T^{1-\d}}{\eta^\d}
    =
    \frac{\d}{8(1+\d)} T
\end{equation*}
which yields
\begin{equation*}
    \tau
    =
    \frac{4(1 + \d)}{\d}\qty( \frac{\d^2}{8}  - \order{\d^3 \log\frac{1}{\d}} )
    \frac{T^{1 - \d}}{\eta^\d}
    =
    \qty( \frac{\d}{2}  - \order{\d^2 \log\frac{1}{\d}} )
    \frac{T^{1 - \d}}{\eta^\d}
\end{equation*}

We note that this indeed yields $\tau = o(T)$ under our previous assumption that $\eta T = \omega(1)$.
We now calculate the optimizer's utility in the last $T/2 + \tau$ rounds:
\begin{alignat*}{3}
    \Line{
        \sum_{t = T/2 - \tau + 1}^T \! \UO\qty(\frac{1}{\l^{(t)}\mu})
    }{=}{
        \frac{\d^\d}{2(1+\d)^\d} \sum_{t = T/2 - \tau + 1}^T \qty( \l^{(t)} )^{-\d}
    }{}
    \\
    \Line{}{\ge}{
        \frac{\d^\d}{2(1+\d)^\d}
        \sum_{t = 1}^{T/2 + \tau} \qty( 
            1
            +
            t
            \frac{\eta}{2}
            \qty(1 - \frac{\d}{2} - \order{\d^2 \log\frac{1}{\d}})
        )^{-\d}
    }{}
    \\
    \Line{}{\ge}{
        \frac{\d^\d}{2(1+\d)^\d}
        \int_{t = 1}^{T/2 + \tau} \qty( 
            1
            +
            t
            \frac{\eta}{2}
            \qty(1 - \frac{\d}{2} - \order{\d^2 \log\frac{1}{\d}})
        )^{-\d}
    }{}
    \\
    \Line{}{=}{
        \frac{\d^\d}{2(1+\d)^\d}
        \frac{
            \qty( 
                1
                +
                \qty( \frac{T}{2} + \tau )
                \frac{\eta}{2}
                \qty(1 - \frac{\d}{2} - \order{\d^2 \log\frac{1}{\d}})
            )^{1-\d}
            \!\!\!-
            \qty( 
                1
                +
                \frac{\eta}{2}
                \qty(1 - \frac{\d}{2} - \order{\d^2 \log\frac{1}{\d}})
            )^{1-\d}
        }{
            \frac{\eta}{2}
            \qty(1 - \frac{\d}{2} - \order{\d^2 \log\frac{1}{\d}})
            (1 - \d)
        }
    }{}
    \\
    \Line{}{\approx}{
        \qty(
            \frac{1}{4} - \order{\d \log\frac{1}{\d}}
        )
        \frac{T^{1-\d}}{\eta^\d}
    }{}
\end{alignat*}

Since in the first $T/2 - \tau$ rounds the optimizer is getting value $\frac{1}{2}$ per round, the above implies that the total value the optimizer gets is
\begin{equation*}
    \frac{T}{4}
    -
    \frac{\tau}{2}
    +
    \qty(
        \frac{1}{4} - \order{\d \log\frac{1}{\d}}
    )
    \frac{T^{1-\d}}{\eta^\d}
    \approx
    \frac{T}{4}
    +
    \qty( \frac{1}{4}  - \order{\d \log\frac{1}{\d}} )
    \frac{T^{1-\d}}{\eta^\d}
\end{equation*}

If we were to set $\eta = \Theta(T^{-\alpha})$ for some $\alpha \in (0, 1)$ (which satisfies all our previous assumptions), the optimizer's benefit would be $\Omega(T^{1 - (1-\alpha)\d})$.
As $\d \to 0$ this converges to linear in $T$ benefit by switching to bidding $1/\mu$ after round $T/2 - \tau$ compared to the BSE strategy of bidding $1$ for the first $T/2$ rounds.

\subsection{Experimental Analysis} \label{ssec:app:dynamic:experiment}

\begin{figure}[t!]
    \centering
    \includegraphics[width=.47\linewidth]{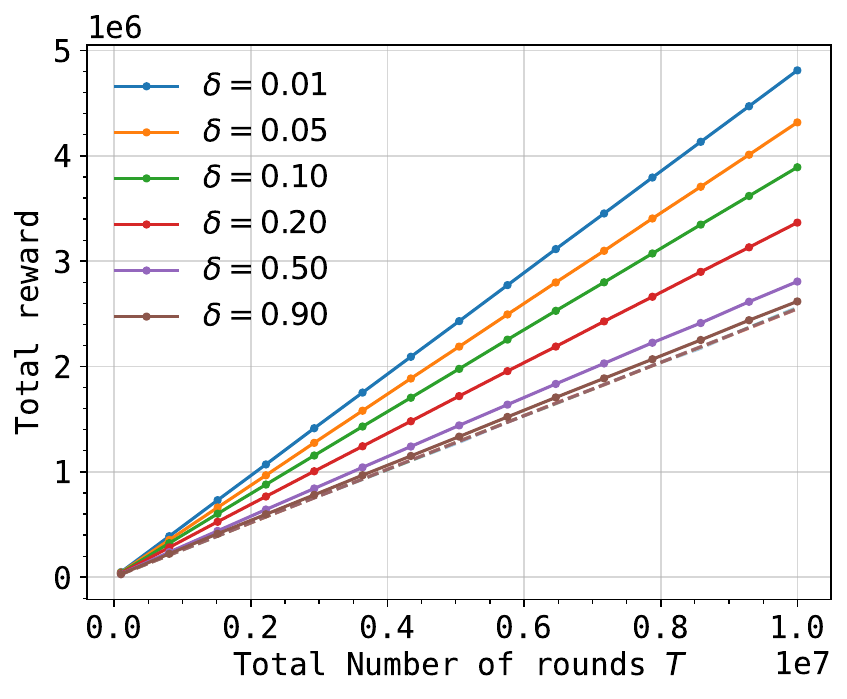}%
    \hfill%
    \includegraphics[width=.5\linewidth]{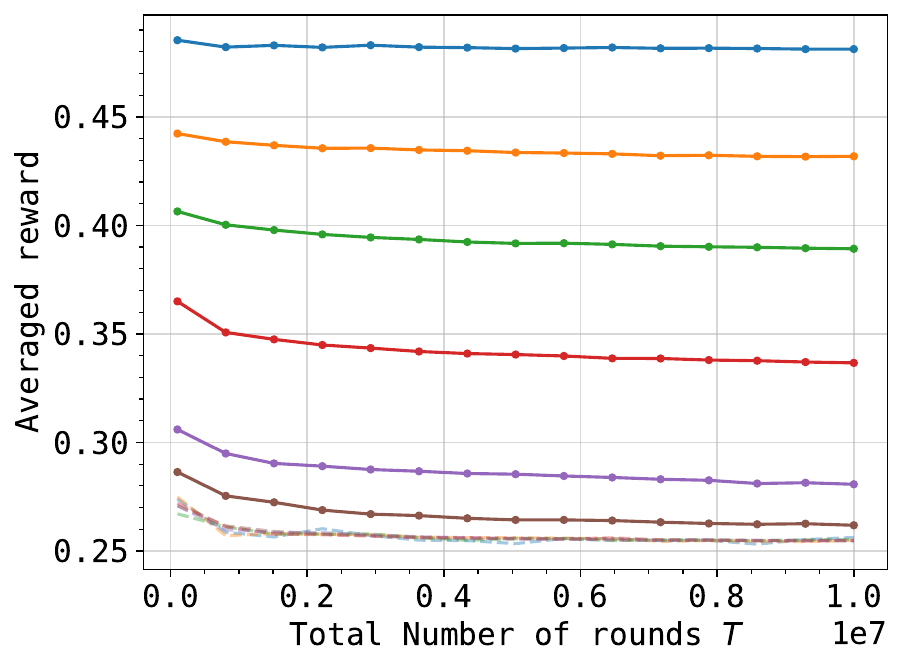}%
    \caption{The experiments of \cref{ssec:app:dynamic:experiment}. For $10^5 \le T \le 10^7$ (x-axis) and different $\d$ (different lines) we compare the optimizer's reward when using strategy \ref{item:app:dynamic:strat1} (dotted lines that collapse to one line) and strategy \ref{item:app:dynamic:strat2} (solid lines).
    }
    \label{fig:app:dynamic:experiment}
\end{figure}

We complement the above analysis of the expected outcome with the following experiment.
First, we fix $\eta = T^{-2/3}$, as in \cref{thm:main}, in which it is guaranteed that the optimizer can get at most $\order*{T^{2/3}}$ above her BSE value.
The learner starts at $\l^{(1)} = 0$ and follows the update \cref{eq:update}.
We then compare the following two strategies by the Optimizer:
\begin{enumerate}
    \item\label{item:app:dynamic:strat1}
    The optimizer bids $\l^{(t)}$ in round $t$, as long as there is budget remaining.

    \item\label{item:app:dynamic:strat2}
    The optimizer bids $\l^{(t)}$ in round $t$, for $t \le \frac{T}{2} - \tau$ and then bids $\frac{1}{\mu} = \frac{\d}{2(1+\d)}$.
    If at any point there is no budget remaining, bid $0$.
\end{enumerate}

We notice that bidding $\l^{(t)}$ and then stopping after half the rounds is the BSE.
In fact, because the learner starts at a low $\l^{(1)} = 0$ and takes $\Theta(1/\eta)$ rounds to converge to $\l = 1$, we expect strategy \ref{item:app:dynamic:strat1} to do slightly higher than the BSE value of $\frac{T}{4}$.
We show the value that the optimizer gets for different values of $T$ and $\d$.
Specifically, we consider $10^5 \le T \le 10^7$ and $0.01 \le \d \le 0.9$.

We present the results of the experiments in \cref{fig:app:dynamic:experiment}.
We make multiple observations.
\begin{enumerate}
    \item Due to the large number of rounds, there is little variation between experiments, implying that random events concentrate very well.

    \item For all the values of $\d$, strategy \ref{item:app:dynamic:strat1} does slightly better than the BSE value of $\frac{T}{4}$.
    In fact, for all the values of $\d$ we examine, this strategy does consistently the same.

    \item Strategy \ref{item:app:dynamic:strat2} does considerably better.
    Specifically, according to our previous analysis and since $\eta = T^{-2/3}$, we expect the optimizer's reward to be $\frac{T}{4}\qty\big(1 + T^{-\d/3})$ for small $\d$.
    For $\d = 0.01$, this reward becomes $\frac{T}{4}\qty\big(1 + T^{-0.003})$.
    For the range $10^5 \le T \le 10^7$, this makes the optimizer's average reward $\frac{1}{4}\qty\big(1 + T^{-\d/3}) \in [0.487, 0.491]$.
    This is very close to double the BSE value of $T/4$ and also very close to what we get experimentally.
\end{enumerate}

\section{Manipulability of Online Gradient Descent}
\label{sec:app:OGD_manipulable}

Much like our pacing algorithm~\eqref{eq:update}---which acts as a PID-style controller adjusting $\l^{(t)}$ proportionally to $\rL - p_{\mathtt{L}}^{(t)}$ and converging within $\order*{1/\eta}$ rounds on any stationary phase---Projected Online Gradient Descent (OGD) with fixed step size $\eta$ is a PID-style controller on the learner's action simplex that adapts within $\order*{1/\eta}$ rounds on stationary phases and guarantees $\order*{\sqrt{T}}$ adaptive regret on every interval~\cite{Hazan16} (see \cref{remark:defs:opt_achievable,rem:auctions:BSE_achievable}).
Here we present a simple $3 \times 3$ unbudgeted game ($\PO = 0$) showing that fast adaptation alone does not prevent manipulation: by switching actions whenever the learner reaches a simplex corner, the optimizer keeps OGD permanently in transition around the boundary of $\Delta([3])$ and extracts $\Omega(T)$ more utility than the Stackelberg equilibrium value $\SE$.
We do not claim that this result is new; OGD can have $\Omega(T)$ swap regret and by~\citet{MMSS22} there is an optimizer who can get $\Omega(T)$ more utility than the $\SE$ value.
We include this example to illustrate how such a process can be manipulated.

\paragraph{Game definition and Stackelberg value}
Consider an unbudgeted game ($\PO = \PL = 0$, where the BSE of \cref{def:defs:budgeted_stackelberg} reduces to the classic Stackelberg equilibrium; see \cref{sec:def}) with action sets $\calA = \calB = [3]$ and payoff matrices
\begin{equation}
    \label{eq:app:OGD:payoffs}
    \UO = \m{
        1 & 0 & 0 \\
        0 & 1 & 0 \\
        0 & 0 & 1
    },
    \qquad
    \UL = \m{
        0 & 1 & \nicefrac{1}{2} \\
        \nicefrac{1}{2} & 0 & 1 \\
        1 & \nicefrac{1}{2} & 0
    }.
\end{equation}
For an optimizer mixed strategy $x = (x_1, x_2, x_3)^\top \in \Delta([3])$, the learner's action utilities are $x^\top \UL = (\frac{1}{2}x_2 + x_3,\; x_1 + \frac{1}{2}x_3,\; \frac{1}{2}x_1 + x_2)$.
Pure action $1$ beats action $2$ iff $\frac{1}{2}x_2 + x_3 \ge x_1 + \frac{1}{2}x_3$, i.e., $x_1 \le \frac{1}{2}(1 - x_1) \iff x_1 \le \nicefrac{1}{3}$.
By cyclic symmetry, every $j \in \textrm{supp}(y)$ for $y \in \BR(x)$ satisfies $x_j \le \nicefrac{1}{3}$, so $x^\top \UO y = \sum_{j \in [3]} y_j x_j \le \nicefrac{1}{3}$.
Since $x^\star = (\nicefrac{1}{3}, \nicefrac{1}{3}, \nicefrac{1}{3})^\top$ achieves $\nicefrac{1}{3}$, we have $\SE = \nicefrac{1}{3}$.

\begin{figure}[t!]
    \centering
    \includegraphics[width=\linewidth]{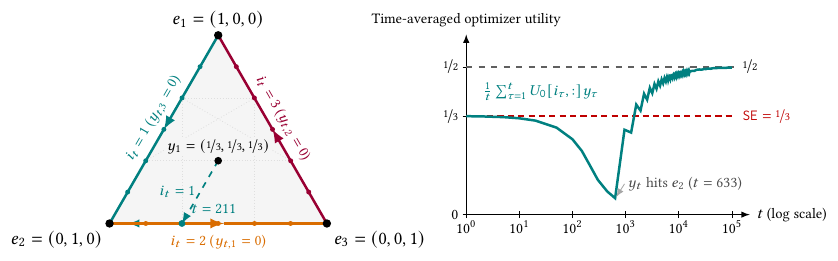}
    \caption{Simulation of Projected OGD in the $3 \times 3$ game of \cref{eq:app:OGD:payoffs} over $T = 100{,}000$ rounds with step size $\eta = 1/\sqrt{T}$. \textbf{Left:} Trajectory of $y_t \in \Delta([3])$ from $y_1 = (\nicefrac{1}{3}, \nicefrac{1}{3}, \nicefrac{1}{3})$ to the boundary (dashed) and around $e_2 \to e_3 \to e_1 \to e_2$ (solid). \textbf{Right:} Optimizer's running average utility $\frac{1}{t}\sum_{\tau=1}^t \UO[i_\tau, :] y_\tau$ converging to $\nicefrac{1}{2} > \SE = \nicefrac{1}{3}$.}
    \label{fig:app:OGD_manipulable}
\end{figure}

\paragraph{Manipulation by corner switching}
In round $t$, after observing $i_t \in [3]$, the learner updates $y_t \in \Delta([3])$ via Projected OGD on the loss $-\UL[i_t, :] y$ from $y_1 = (\nicefrac{1}{3}, \nicefrac{1}{3}, \nicefrac{1}{3})$:
\begin{equation}
    \label{eq:app:OGD:update}
    y_{t+1} = \Pi_{\Delta([3])}\qty(y_t + \eta \, \UL[i_t, :]^\top),
    \qquad\text{ where }\qquad
    \Pi_{\Delta([3])}(z) = \argmin_{y \in \Delta([3])} \norm{y - z}_2^2.
\end{equation}

Against any fixed $x$ that admits a unique best response, $y_t$ reaches the corresponding corner in $\order*{1/\eta}$ rounds, and for $\eta = \Theta(1/\sqrt{T})$ OGD satisfies $$\sum_{t \in \mathcal{I}} \UL[i_t, :]\qty(y^\star - y_t) \le \frac{1}{\eta} + \frac{5\eta}{8}\abs{\mathcal{I}} = \order*{\sqrt{T}}$$ on every contiguous interval $\mathcal{I} \subseteq [T]$ and $y^\star \in \Delta([3])$~\cite{Hazan16}.
Nevertheless, the optimizer can exploit OGD (since it can incur $\Omega(T)$ swap regret~\cite{MMSS22}) by playing action $1$ until $y_t$ hits $e_2 = (0, 1, 0)$, then switching to action $2$ until $y_t = e_3 = (0, 0, 1)$, then action $3$ until $y_t = e_1 = (1, 0, 0)$, and repeating cyclically ($e_i \to e_{(i \bmod 3) + 1}$).

Starting from $y_t = e_1 = (1, 0, 0)$ under $i_t = 1$ (payoff vector $u^{(1)} = (0, 1, \nicefrac{1}{2})^\top$), at any $y_t = (y_{t,1}, 1 - y_{t,1}, 0)$ on edge $e_1 e_2$ the projection $\Pi_{\Delta([3])}(y_t + \eta u^{(1)})$ subtracts $\frac{\eta}{2}$ from each coordinate and clips at zero.
Since action $3$'s increment equals $\frac{\eta}{2}$, $y_{t+1,3} = 0$ stays pinned at zero while $y_{t,1}$ drops from $1$ to $0$ in steps of $\frac{\eta}{2}$ over $\ceil{2/\eta}$ rounds, giving the optimizer utility $y_{t,1}$ per round---an average of $\nicefrac{1}{2}$.
By cyclic symmetry, every edge yields average utility $\nicefrac{1}{2}$, for a total of $\frac{1}{2}T - \order*{1/\eta} = \SE \cdot T + \frac{1}{6}T - \order*{\sqrt{T}}$ when $\eta = 1/\sqrt{T}$.
Unlike $\Delta([3])$, where $y_t$ can traverse a closed loop $e_2 \to e_3 \to e_1 \to e_2$ without reversing, our pacing state $\l^{(t)} \in \Rz$ is one-dimensional, allowing the scalar potential $G_\sigma(\l)$ of \cref{sec:warmup,sec:dynamic} to rule out cyclic exploitation.

\paragraph{Simulation}
Simulating \cref{eq:app:OGD:update} for $T = 100{,}000$ and $\eta = 1/\sqrt{T}$ (\cref{fig:app:OGD_manipulable}), $y_t$ hits $y_{t,1} = 0$ at $t = 211$ and $e_2$ at $t = 633$, then cycles along $e_2 \to e_3 \to e_1 \to e_2$ in $\ceil{2/\eta} = 633$ rounds per edge ($157$ switches), yielding time-averaged optimizer utility $0.498 \approx \nicefrac{1}{2} > \SE = \nicefrac{1}{3}$.

\end{document}